\documentclass{amsart}

\newtheorem{definition}{Definition}
\newtheorem{proposition}{Proposition}
\newtheorem{theorem}{Theorem}
\newtheorem{lemma}{Lemma}
\newtheorem{corollary}{Corollary}
\newtheorem{remark}{Remark}
\newtheorem{assumption}{Assumption}

\usepackage{graphicx}
\usepackage{cite}
\usepackage{subfig}
\usepackage{float}

\numberwithin{equation}{section}

\begin{document}

\title{THE WISHART SHORT-RATE MODEL}

\author{Alessandro Gnoatto}

\address{Mathematical Institute, Ludwig-Maximilians-University,\\
Theresienstrasse 39, 80333 Munich Germany and\\
 Department of Mathematics, University of Padova,\\
  Via Trieste 63, 35121 Padova, Italy }
  \email{gnoatto@mathematik.uni-muenchen.de}

\thanks{The author is grateful to Christa Cuchiero, Martino Grasselli, Martin Keller-Ressel, Wolfgang Runggaldier and Josef Teichmann.\\
Electronic version of an article published as [International Journal of Theoretical and Applied Finance, 15, 8, 2012] [DOI 10.1142/S0219024912500562] \copyright [ World Scientific Publishing Company] [http://http://www.worldscientific.com/worldscinet/ijtaf]}

%\begin{history}
%\received{(4 February 2012)}
%\revised{(19 June 2012)}
%\accepted{(17 September 2012)}
%\end{history}

\begin{abstract}
We consider a short rate model, driven by a stochastic process on the cone of positive semidefinite matrices. We derive sufficient conditions ensuring that the model replicates normal, inverse or humped yield curves.
\end{abstract}

\maketitle
%\keywords{Yield curve shapes, Wishart processes, Affine processes, Zero-coupon bond}

\section{Introduction}
In the present paper we focus on models where the short rate is given as
\begin{equation}
r_t=a+Tr\left[BX_t\right],\label{shortRate}
\end{equation}
where $a\in\mathbb{R}_{\geq 0}$, $B$ is a symmetric positive definite matrix and $X=\left(X_t\right)_{t\geq0}$ is a stochastic process on the cone of positive semidefinite matrices.
We will provide a fairly general pricing formula for zero coupon bonds for this family of models. Then we will restrict to the \textbf{Wishart short rate model}. This kind of model has been suggested in \cite{article_gou02} and \cite{article_Grasselli}, then \cite{article_BCT} investigated the properties of this model with respect to many issues concerning the yield curve and interest rate derivatives. An analysis of the impact of the specification of the risk-premium is provided in \cite{article_Chiarella}.

We provide a set of sufficient conditions ensuring that this short rate model produces certain shapes of the yield curve. Our analysis of yield curve shapes is inspired by the work of \cite{article_KRS}, where a set of restrictions on the shapes of the yield curve is derived, under the assumption that the driving process is affine in the sense of \cite{article_DFS}. 

Affine processes have been applied in finance in many contexts, for an application to interest rates, see \cite{article_DuffieKan}. More recently \cite{article_DFS} and \cite{thesis_KR} provided a full characterization for the state space $\mathbb{R}_{\geq 0}^{m} \times \mathbb{R}^{n}$. In \cite{article_Cuchiero}, an analogous characterization is provided for the state space $S_{d}^{+}$. An alternative characterization is proposed in \cite{article_Grasselli}, where the concept of solvable affine term structure model is discussed, both for the state space $\mathbb{R}_{\geq 0}^{m} \times \mathbb{R}^{n}$ and $S_d^+$.

Affine processes on positive semidefinite matrices have witnessed an increasing importance in applications in finance. The first application was presented in \cite{article_gou02} and then \cite{article_gou03}. Applications to multi-factor stochastic volatility, and stochastic correlations may be found in \cite{article_DaFonseca1}, \cite{article_DaFonseca2}, \cite{article_DaFonseca3}, \cite{article_DaFonseca4}, \cite{article_gou01} and \cite{article_BPT}, for the pure diffusion case. \cite{article_LT} introduce a class of jump diffusions where the intensity is an affine function of the state variable. Jump processes on $S_d^+$ are treated in \cite{article_BNS}, \cite{article_BNS02}, \cite{article_MPS}, \cite{article_MPS02} and \cite{article_PS}.

The present paper is organized as follows: first we introduce the general setup of affine processes on the cone state space $S_d^+$, thus providing the general setup for the pricing problem. Then we restrict to the Wishart process and provide a closed form formula for zero coupon bonds due to \cite{article_gou03}. Finally, by assuming that the Wishart process lies in the interior of $S_d^+$, we are also able to provide a set of sufficient conditions on the initial state of the process, which ensure that the model replicates certain shapes of the yield curve.

\section{Affine Processes on $S_{d}^{+}$}
\subsection{General Results}
For the reader's convenience, we report here some results which may be found in \cite{article_Cuchiero}, which constitute the theoretical framework we will be working with.
Let $(\Omega, \mathcal{F} ,\left(\mathcal{F}_t \right)_{t \geq 0}, \mathbb{P})$ be a filtered probability space, with the filtration $\left(\mathcal{F}_t \right)_{t \geq 0}$ satisfying the usual assumptions. Let $S_{d}^{+}$ denote the cone of positive semidefinite $d \times d$ matrices, endowed with the scalar product $\left\langle x,y\right\rangle=Tr\left[ xy\right]$.
We consider a Markov process $X=\left(X_{t} \right)_{t \geq 0}$ with state space $S_{d}^{+}$, transition probability $p_{t}(X_0,A)=\mathbb{P}(X_{t} \in A)$ for $A \in S_{d}^{+}$, and transition semigroup $\left( P_{t}\right)_{t \geq 0}$ acting on bounded functions $f \in S_{d}^{+}$:

\begin{definition}
(\cite{article_Cuchiero} Definition 2.1) \label{def1} The Markov process $X$ is called affine if:
\begin{enumerate}
\item it is stochastically continuous, that is, $\lim_{s \to t}p_{s}(X_0, \cdot)=p_{t}(X_0, \cdot)$ weakly on $S_{d}^{+}$ $\forall t \text{, } X_0 \in S_{d}^{+}$, and\\
\item its Laplace transform has exponential-affine dependence on the initial state:
\begin{equation}
P_{t}e^{-Tr\left[uX_0 \right]}=\int_{S_{d}^{+}}{e^{-Tr\left[u\xi \right]}p_t(X_0,d\xi)}=e^{-\phi(t,u)-Tr\left[ \psi(t,u)X_0\right]},\label{defAP}
\end{equation}
$\forall t\text{ and }u, X_{0} \in S_{d}^{+}$, for some functions $\phi: \mathbb{R}_{\geq 0}\times S_{d}^{+} \rightarrow \mathbb{R}_{\geq 0}$ and $\psi: \mathbb{R}_{\geq 0}\times S_{d}^{+} \rightarrow S_{d}^{+}$.
\end{enumerate}
\end{definition}

Note that in the definition above we assumed that the process is stochastically continuous, a feature that implies, according to Proposition 3.4 in \cite{article_Cuchiero}, that the process is regular in the following sense:

\begin{definition} (\cite{article_Cuchiero} Definition 2.2)
The affine process $X$ is called regular if the derivatives
\begin{equation}
F(u)=\frac{\partial \phi(t,u)}{\partial t}|_{t=0^+}, \hspace{10mm}R(u)=\frac{\partial \psi(t,u)}{\partial t}|_{t=0^+}
\end{equation}
exist and are continuous at $u=0$.
\end{definition}

\begin{definition}
(\cite{article_Cuchiero} Definition 2.3) Let $\chi:S_d^+\rightarrow S_d^+$ be some bounded continuous truncation function with $\chi(\xi)=\xi$ in a neighborhood of 0.  An admissible parameter set $\left(\alpha,b,\beta^{ij},c,\gamma,m,\mu \right)$ associated with $\chi$ consists of:
\begin{itemize}
\item a linear diffusion coefficient
\begin{equation}
\alpha \in S_{d}^{+},
\end{equation}
\item a constant drift term
\begin{equation}
b\succeq (d-1)\alpha,
\end{equation}
\item a constant killing rate term
\begin{equation}
c \in \mathbb{R}^{+},
\end{equation}
\item a linear killing rate coefficient
\begin{equation}
\gamma \in S_{d}^{+},
\end{equation}
\item a constant jump term: a Borel measure $m$ on $S_{d}^{+} \setminus \left\{ 0\right\}$ satisfying
\begin{equation}
\int_{S_{d}^{+} \setminus \left\{ 0\right\}}{\left(\parallel \xi \parallel \land 1 \right)m(d \xi)}<\infty,
\end{equation}
\item a linear jump coefficient: a $d \times d$ matrix $\mu=(\mu_{ij})$ of finite signed measure on $S_{d}^{+} \setminus \left\{ 0\right\}$, such that $\mu(E)\in S_{d}^{+}$ $\forall E \in \mathcal{B}(S_{d}^{+})$ and the kernel
\begin{equation}
M(x,d\xi):=\frac{Tr\left[x \mu(d\xi) \right]}{\parallel \xi \parallel^2 \land 1}\label{jump_meas}
\end{equation}
satisfies
\begin{equation}
\int_{S_{d}^{+} \setminus \left\{ 0\right\}}{Tr\left[\chi(\xi)u \right]M(x,d\xi)}<\infty, 
\end{equation}
$\forall x,u \in S_{d}^{+}\text{ s.t. }Tr\left[xu\right]=0$.
\item a linear drift coefficient: a family $\beta^{ij}=\beta^{ji} \in S_{d}^{+}$ s.t. the linear map $\beta:S_{d} \rightarrow S_{d}$ of the form 
\begin{equation}
\beta(x)=\sum_{i,j}{\beta^{ij}x_{ij}},\label{adm_drift}
\end{equation}
satisfies
\begin{align}
&Tr\left[\beta(x)u\right]-\int_{S_{d}^{+} \setminus \left\{ 0\right\}}{Tr\left[\chi(\xi)u \right]M(x,d\xi)}\geq 0
\end{align}
$\forall x,u \in S_{d}^{+} \text{ with } Tr\left[xu\right]=0$.
\end{itemize}
\end{definition}
The following theorem closes our survey on affine processes. It is a generalization of the result by \cite{article_DFS} to the state space $S_{d}^{+}$. Denote by $\mathcal{S}_{d}$ the space of rapidly decreasing $C^\infty$-functions on $S_d^+$ and let $\mathcal{D}(\mathcal{A})$ be the domain of the generator of the process.

\begin{theorem}
(\cite{article_Cuchiero} Theorem 2.4) Suppose $X$ is an affine process on $S_{d}^{+}$.\label{cuch_2.4} Then $X$ is regular and has the Feller property. Let $\mathcal{A}$ be its infinitesimal generator on $C_{0}(S_{d}^{+})$. Then $\mathcal{S}_{d} \subset \mathcal{D}(\mathcal{A})$ and there exists an admissible parameter set $\left(\alpha,b,\beta^{ij},c,\gamma,m,\mu \right)$ such that, for $f \in \mathcal{S}_{d}$
\begin{align}
\mathcal{A}f(x)&=\frac{1}{2}\sum_{i,j,k,l}{A_{ijkl}(x)\frac{\partial^{2}f(x)}{\partial x_{ij} \partial x_{kl}}}+\sum_{i,j}\left(b_{ij}+\beta_{ij}(x) \right)\frac{\partial f(x)}{\partial x_{ij}}-\left(c+Tr\left[ \gamma x\right] \right)f(x)\nonumber\\
&+\int_{S_{d}^{+} \setminus \left\{ 0\right\}}{\left(f(x+\xi)-f(x) \right)m(d\xi)}\nonumber\\
&+\int_{S_{d}^{+} \setminus \left\{ 0\right\}}{\left(f(x+\xi)-f(x)-Tr\left[\chi(\xi)\nabla f(x)\right] \right)M(x,d\xi)}\label{infGen}
\end{align}
where $\beta(x)$ is defined by \eqref{adm_drift}, $M(x,d\xi)$ by \eqref{jump_meas} and
\begin{equation}
A_{ijkl}(x)=x_{ik}\alpha_{jl}+x_{il}\alpha_{jk}+x_{jk}\alpha_{il}+x_{jl}\alpha_{ik}
\end{equation}
Moreover, $\phi(t,u)$ and $\psi(t,u)$ in Definition 2 solve the generalized Riccati differential equations, for $u \in S_{d}^{+}$,
\begin{align}
\frac{\partial \phi(t,u)}{\partial t}&=F(\psi(t,u)), \hspace{10mm} \phi(0,u)=0,\label{ODE_phi}\\
\frac{\partial \psi(t,u)}{\partial t}&=R(\psi(t,u)), \hspace{10mm} \psi(0,u)=u,\label{ODE_psi}
\end{align}
with
\begin{align}
F(u)&=Tr\left[bu\right]+c-\int_{S_{d}^{+} \setminus \left\{ 0\right\}}{\left( e^{-Tr\left[u\xi\right]}-1\right)m(d\xi)},\\
R(u)&=-2u\alpha u+ \beta^{T}(u)+\gamma\nonumber\\
&-\int_{S_{d}^{+} \setminus \left\{ 0\right\}}{\left(\frac{e^{-Tr\left[u\xi\right]}-1-Tr\left[\chi(\xi)u\right]}{\parallel \xi \parallel^2 \land 1}\right)\mu(d\xi)},
\end{align}
where $\beta^{T}_{ij}(u)=Tr\left[\beta^{ij}u\right]$.\\
Conversely, let $\left(\alpha,b,\beta^{ij},c,\gamma,m,\mu \right)$ be an admissible parameter set. Then there exists a unique affine process on $S_{d}^{+}$ with infinitesimal generator given by \eqref{infGen} and such that the affine transform formula \eqref{defAP} holds for all $(t,u)\in \mathbb{R}_{\geq 0} \times S_{d}^{+}$, where $\phi(t,u)$ and $\psi(t,u)$ are given by \eqref{ODE_phi} and \eqref{ODE_psi}.
\end{theorem}

\subsection{Bond Prices}
In this section we derive a fairly general pricing formula for zero coupon bonds. Before this, we would like to spend a couple of words to recall an important fact concerning the risk neutral measure that we will use for pricing purposes. From  \cite{bjork} we know that it is quite tempting to consider the short rate as a traded asset and treat zero coupon bonds as derivatives written on the short rate. Unfortunately, the short rate is not a traded asset, hence the bond market is arbitrage free but not complete. This means that in general there exist many risk neutral measures. This implies that the reference risk neutral measure $\mathbb{Q}$ will be inferred in general from market prices, and so will result from a calibration procedure.
Let $B \in S_{d}^{+}$ then, according to Definition \ref{def1}, we have:
\begin{equation}
\mathbb{E}^{\mathbb{Q}}\left[e^{-Tr\left[BX_t\right]}\right]=e^{-\phi(t,B)-Tr\left[\psi(t,B)x\right]}.
\end{equation}
More generally, for $t,s,>0$:
\begin{equation}
\mathbb{E}^{\mathbb{Q}}\left[e^{-Tr\left[BX_{t+s}\right]}|\mathcal{F}_t\right]=e^{-\phi(s,B)-Tr\left[\psi(s,B)x_t\right]}.
\end{equation}
%\begin{equation}
%P(\tau)=e^{-\tilde{\phi}(\tau,B)-Tr\left[\tilde{\psi}(\tau,B)x_t\right]}.
%%\mathbb{E}^{\mathbb{Q}}\left[e^{-\int_{t}^{T}{a+Tr\left[BX_{u}\right]du}}|\mathcal{F}_t\right]
%\end{equation}
In what follows, by defining $\tau=T-t$, we will see that a similar formula holds for the price of a zero coupon bond which is computed, when the short rate is given as in \eqref{shortRate}, via the following expectation:
\begin{equation}
P_t(\tau):=\mathbb{E}^{\mathbb{Q}}\left[e^{-\int_{t}^{T}{a+Tr\left[BX_{u}\right]du}}|\mathcal{F}_t\right].
\end{equation}
This expectaton satisfies the following Kolmogorov backward equation:
\begin{equation}
\frac{\partial P_t}{\partial \tau}=\mathcal{A}P_t-\left(a+Tr\left[BX\right]\right)P_t,\quad P_t(0)=1,
\end{equation}
where the infinitesimal generator of the process $X$ is reported in Theorem \ref{cuch_2.4}.
We introduce an exponentially affine guess given by the following:
\begin{align}
P_t(\tau)=\exp\left\{-\tilde{\phi}(\tau,B)-Tr\left[\tilde{\psi}(\tau,B)X\right]\right\},\label{exp_guess}
\end{align}
so that:
\begin{equation}
\frac{\partial P_t}{\partial \tau}=\left(-\frac{\partial \tilde{\phi}}{\partial \tau}-Tr\left[\frac{\partial \tilde{\psi}}{\partial \tau}X \right] \right)P_t,
\end{equation}
\begin{equation}
\mathcal{A}e^{-\tilde{\phi}(\tau,B)-Tr\left[\tilde{\psi}(\tau,B)X \right]}=e^{-\tilde{\phi}(\tau,B)}\mathcal{A}e^{-Tr\left[\tilde{\psi}(\tau,B)X \right]},
\end{equation}
and (see always Theorem \ref{cuch_2.4}):
\begin{align}
&\mathcal{A}e^{-Tr\left[\tilde{\psi}(\tau,B)X \right]}=\left(-F(\tilde{\psi}(\tau,B))-Tr\left[R(\tilde{\psi}(\tau,B))X\right]\right)e^{-Tr\left[\tilde{\psi}(\tau,B)X \right]}\nonumber\\
&=\left\{-Tr\left[b\tilde{\psi}(\tau,B)\right]+\int_{S_{d}^{+} \setminus \left\{ 0\right\}}{\left(e^{-Tr\left[\tilde{\psi}(\tau,B)\xi\right]}-1 \right)m(d\xi)}\right.\nonumber\\
&+Tr\Bigg[\Bigg(2\tilde{\psi}(\tau,B)\alpha\tilde{\psi}(\tau,B)-\beta^{\top}(\tilde{\psi}(\tau,B))\Bigg.\Bigg.\nonumber\\
&+\left.\left.\Bigg.\int_{S_{d}^{+} \setminus \left\{ 0\right\}}{\frac{e^{-Tr\left[\tilde{\psi}(\tau,B)\xi\right]}-1+Tr\left[\chi(\xi)\tilde{\psi}(\tau,B)\right]}{\parallel \xi \parallel^2 \land 1}\mu(d\xi)}\Bigg)X\right]\right\}\nonumber\\
&\times e^{-Tr\left[\tilde{\psi}(\tau,B)X \right]}.
\end{align}
In summary, we obtain:
\begin{align}
&-\frac{\partial \tilde{\phi}}{\partial \tau}-Tr\left[\frac{\partial \tilde{\psi}}{\partial \tau}X \right]=-Tr\left[b\tilde{\psi}(\tau,B)\right]+\int_{S_{d}^{+} \setminus \left\{ 0\right\}}{\left(e^{-Tr\left[\tilde{\psi}(\tau,B)\xi\right]}-1 \right)m(d\xi)}\nonumber\\
&+Tr\Bigg[\Bigg(2\tilde{\psi}(\tau,B)\alpha\tilde{\psi}(\tau,B)-\beta^{\top}(\tilde{\psi}(\tau,B))\Bigg.\Bigg.\nonumber\\
&\left.\Bigg.+\int_{S_{d}^{+} \setminus \left\{ 0\right\}}{\frac{e^{-Tr\left[\tilde{\psi}(\tau,B)\xi\right]}-1+Tr\left[\chi(\xi)\tilde{\psi}(\tau,B)\right]}{\parallel \xi \parallel^2 \land 1}\mu(d\xi)}\Bigg)X\right]\nonumber\\
&-\left(a+Tr\left[BX\right]\right).
\end{align}
Identify terms to obtain the system of (matrix) ODE's:
\begin{align}\label{ODE_bond_phi}
\frac{\partial \tilde{\phi}}{\partial \tau}&=\mathcal{F}\left(\tilde{\psi}(\tau,B)\right)=Tr\left[b\tilde{\psi}(\tau,B)\right]\nonumber\\
&-\int_{S_{d}^{+} \setminus \left\{ 0\right\}}{\left(e^{-Tr\left[\tilde{\psi}(\tau,B)\xi\right]}-1 \right)m(d\xi)}+a,
\end{align}
\begin{align}
\frac{\partial \tilde{\psi}}{\partial \tau}&=\mathcal{R}\left(\tilde{\psi}(\tau,B)\right)=-2\tilde{\psi}(\tau,B)\alpha\tilde{\psi}(\tau,B)+\beta^{\top}(\tilde{\psi}(\tau,B))\nonumber\\
&-\int_{S_{d}^{+} \setminus \left\{ 0\right\}}{\frac{e^{-Tr\left[\tilde{\psi}(\tau,B)\xi\right]}-1+Tr\left[\chi(\xi)\tilde{\psi}(\tau,B)\right]}{\parallel \xi \parallel^2 \land 1}\mu(d\xi)}+B.\label{ODE_bond_psi}
\end{align}
We have thus proven the following:
\begin{proposition}\label{gen_pric}
Let $X$ be a conservative affine process on $S_{d}^{+}$ under the risk neutral probability measure $\mathbb{Q}$. Let the short rate be given as:
\begin{equation}
r_t=a+Tr\left[BX_t\right],
\end{equation}
then the price of a \textit{zero-coupon bond} is given by:
\begin{align}
P_t(\tau):&=\mathbb{E}^{\mathbb{Q}}\left[e^{-\int_{t}^{T}{a+Tr\left[BX_{u}\right]du}}|\mathcal{F}_t\right]\nonumber\\
&=\exp\left\{-\tilde{\phi}(\tau,B)-Tr\left[\tilde{\psi}(\tau,B)X_t\right]\right\},
\end{align}
where $\tilde{\phi}$ and $\tilde{\psi}$ satisfy the following ODE's:
\begin{align}
\frac{\partial \tilde{\phi}}{\partial \tau}&=\mathcal{F}\left(\tilde{\psi}(\tau,B)\right)=Tr\left[b\tilde{\psi}(\tau,B)\right]-\int_{S_{d}^{+} \setminus \left\{ 0\right\}}{\left(e^{-Tr\left[\tilde{\psi}(\tau,B)\xi\right]}-1 \right)m(d\xi)}+a,\nonumber\\
\tilde{\phi}(0,B)&=0,
\end{align}
\begin{align}
\frac{\partial \tilde{\psi}}{\partial \tau}&=\mathcal{R}\left(\tilde{\psi}(\tau,B)\right)=-2\tilde{\psi}(\tau,B)\alpha\tilde{\psi}(\tau,B)+\beta^{\top}(\tilde{\psi}(\tau,B))\nonumber\\
&-\int_{S_{d}^{+} \setminus \left\{ 0\right\}}{\frac{e^{-Tr\left[\tilde{\psi}(\tau,B)\xi\right]}-1+Tr\left[\chi(\xi)\tilde{\psi}(\tau,B)\right]}{\parallel \xi \parallel^2 \land 1}\mu(d\xi)}+B,\label{bond_sol_2}\nonumber\\
\tilde{\psi}(0,B)&=0.
\end{align}
\end{proposition}
\begin{remark}
In summary, with respect to the case where we are interested in the Laplace transform of the process (Theorem \ref{cuch_2.4}), we notice that when we consider the integrated process we have the following:
\begin{align}
\mathcal{F}(u)&:=F(u)+a,\\
\mathcal{R}(u)&:=R(u)+B.
\end{align}
\end{remark}
In the next sections, we will be working repeatedly with the functions $\mathcal{F}(u)$ and $\mathcal{R}(u)$ defined above. They will permit us to characterize in a very precise way the asymptotic behavior of the yield curve for large maturities and they will be a key ingredient to derive our sufficient conditions concerning the shapes of the curve.

\section{The Wishart short rate model}
\subsection{Some Properties of the Matrix Exponential}
We recall some background on the matrix exponential, which will be useful in the sequel. First, we provide the following:
\begin{definition}
Let $A$ be a matrix with entries in $\mathbb{C}$, then we define:
\begin{equation}
e^{A\tau}:=\sum_{k=0}^{\infty}{\frac{A^k\tau^k}{k!}}.
\end{equation}
\end{definition}
In the sequel we will look at the asymptotic behavior of the yield curve, so the following lemma will be useful:
\begin{lemma}\label{lemma1}
Let $A\in M_{d}$. Assume $\Re(\lambda(A))<0,\quad\forall\lambda\in\sigma(A)$, then:
\begin{equation}
\lim_{\tau\to\infty}{e^{A\tau}}=0\in M_{d\times d}.
\end{equation}
\end{lemma}
\begin{proof}
See e.g. \cite{hojo} Chapter 2.
\end{proof}

This fact allows us to determine the asymptotic behavior of the following functions:
\begin{lemma}\label{lemmaTanh}
Let $O\in S_d^+$, define:
\begin{align}
\sinh(O\tau)=\frac{e^{O\tau}-e^{-O\tau}}{2},\quad \cosh(O\tau)=\frac{e^{O\tau}+e^{-O\tau}}{2}
\end{align}
and
\begin{align}
\tanh(O\tau)=\left(\cosh(O\tau)\right)^{-1}\sinh(O\tau),\quad\coth(O\tau)=\left(\sinh(O\tau)\right)^{-1}\cosh(O\tau),
\end{align}
then
\begin{align}
\lim_{\tau\to\infty}{\tanh(O\tau)}=\lim_{\tau\to\infty}{\coth(O\tau)}=I_d.
\end{align}
\end{lemma}
\begin{proof}
\begin{align}
\lim_{\tau\to\infty}{\tanh(O\tau)}=\lim_{\tau\to\infty}{\left(I_d+e^{-2O\tau}\right)^{-1}\left(I_d-e^{-2O\tau}\right)}=I_d.\label{limitTanh}
\end{align}
The second equality follows along the same lines.
\end{proof}

\subsection{Closed-Form Pricing Formulae in the General Diffusion Model}
In this section we consider a diffusion model for the short rate. The driving process we use was first considered in the seminal paper by \cite{article_Bru}, however, in the present paper, following the standard literature on Wishart process, we will be dealing with a slight generalization. Using the terminology of Bru, we assume that the law of $X_t$ is $WIS_d(x_0, \alpha, M,Q)$ under the risk neutral measure $\mathbb{Q}$. $X_t$ is the solution of the following SDE:
\begin{align}
dX_t=\left(b+MX_t+X_tM^{\top}\right)dt+\sqrt{X_t}dW_tQ+Q^{\top} dW_{t}^{\top}\sqrt{X_t},\label{genWis}
\end{align}
where $M,Q \in GL(d)$, $b=\alpha Q^{\top}Q$. We further assume\footnote{The assumption on $\alpha$ implies that the process lies in the interior of the cone $S_d^+$, that we denote by $S_d^{++}$. This more restrictive assumption is required in order to derive the conditions on the shapes of the yield curve. All bond pricing formulae that we outline in the sequel hold true also for $\alpha> d-1$.} $\alpha\geq d+1$ and $x_0 \in S_d^{++}$. These last assumptions, according to Theorem  2.2 in \cite{article_MPS}, allow us to claim that there exists a strong solution to \eqref{genWis} on the interval $\left[0,\tau_0\right)$, where the stopping time $\tau_0$ is defined as:
\begin{align}
\tau_0=\inf\left\{t\geq0|\det X_t=0\right\}.
\end{align}
Moreover, we have $\tau_0=+\infty$ a.s.
Finally, we notice that, in full analogy with the scalar square root process, the term $-\alpha Q^{\top}Q$ is related to the long term matrix $X_\infty$ via the following Lypaunov equation:
\begin{align}
-\alpha Q^{\top}Q=MX_\infty+X_\infty M^{\top},
\end{align}
so that, for the rest of this paper, we make the following standard assumption in order to grant the mean reverting feature of the process $X_t$.
\begin{assumption}\label{assu_eigen}
We require $\Re\lambda<0$, $\forall \lambda\in\sigma(M)$. This requirement implies the convergence of the improper integral $\alpha\int_{0}^{\infty}{e^{Ms}Q^{\top}Q e^{M^{\top}s}ds}$, which satisfies the equation. We further assume that $\forall \tau\in\mathbb{R}_{\geq 0}\cup\left\{\infty\right\}$ the matrix $M-2Q^{\top}Q\tilde{\psi}(\tau)$ has negative real eigenvalues.
\end{assumption}

For this driving process, the price of a zero coupon bond is computed as follows. We use the shorthand notation $\tilde{\phi}(\tau,B)=\tilde{\phi}(\tau)$ and $\tilde{\psi}(\tau,B)=\tilde{\psi}(\tau)$.

\begin{proposition}\label{mainZCP}
Let the short rate be given as in \eqref{shortRate}, for a process $X_t$ with law $WIS_d(x_0, \alpha, M,Q)$. Let $B\in S_d^{++}$ and set $\tau=T-t$.
%\begin{align*}
%{\cal D}_t=&\left\{ w,v\in S_d : \mathbb{E}_{s_0}^{\mathbb{P}}\left[\exp\left\{-Tr\left[wS_t+\int_{0}^{t}{vS_sds}\right]\right\}\right]<+\infty\right\}.\end{align*}
Then the price of a zero coupon bond is given by:
\begin{align*}
&\mathbb{E}_{X_t}^{\mathbb{Q}}\left[\exp\left\{-a\tau-Tr\left[\int_{t}^{T}{BX_sds}\right]\right\}\right]\nonumber\\
&=\exp\left\{-\tilde{\phi}(\tau)-Tr\left[\tilde{\psi}(\tau)X_t\right]\right\},
\end{align*}
where $\tilde{\psi}(\tau)$ and $\tilde{\phi}(\tau)$ solve the following system of ODE's:
\begin{equation}
\frac{\partial \tilde{\phi}}{\partial \tau}=Tr\left[\alpha Q^{\top}Q\tilde{\psi}(\tau)\right]+a,\quad \tilde{\phi}(0)=0,\label{ODE1}
\end{equation}

\begin{align}
\frac{\partial \tilde{\psi}}{\partial \tau}&=\tilde{\psi}(\tau)M+M^{\top}\tilde{\psi}(\tau)-2\tilde{\psi}(\tau)Q^{\top}Q\tilde{\psi}(\tau)+B, \quad \tilde{\psi}(0)=0.\label{ODE2}
\end{align}
\end{proposition}
\begin{proof} Same arguments as in Proposition \ref{gen_pric}.
\end{proof}

\begin{remark}\label{mathcal_R}
In the present setting we have:
\begin{align}
\mathcal{R}\left(\tilde{\psi}(\tau)\right)&=\tilde{\psi}(\tau)M+M^{\top}\tilde{\psi}(\tau)-2\tilde{\psi}(\tau)Q^{\top}Q\tilde{\psi}(\tau)+B\label{mcr}\\
\mathcal{F}\left(\tilde{\psi}(\tau)\right)&=Tr\left[\alpha Q^{\top}Q\tilde{\psi}(\tau)\right]+a\label{mcF}.
\end{align}
Moreover, a direct substitution of the terminal condition $\tilde{\psi}(0)=0$ implies:
\begin{align}
\mathcal{R}\left(\tilde{\psi}(0)\right)&=B.\label{t_cond}
\end{align}
\end{remark}

The solution of the system of ODE's above may be computed by relying on the different approaches which are summarized in the following proposition. Let us denote by $\psi^\prime$ a solution to the algebraic Riccati equation 
\begin{align}
\psi^\prime M+M^{\top}\psi^\prime-2\psi^\prime Q^{\top}Q\psi^\prime+B=0.
\end{align}

\begin{proposition}
The system of ODE's \eqref{ODE1}, \eqref{ODE2} admits the following solution:
%\begin{itemize}
%\item \textbf{Matrix Cameron Martin approach (See \cite{GnoGra2011})}
%\begin{align}
%\tilde{\phi}(\tau)&=-\frac{\alpha}{2}\log\det\left(e^{-M\tau}\left(\cosh(\sqrt{\bar{v}}\tau)+\sinh(\sqrt{\bar{v}}\tau)k\right)\right)+a\tau,\\
%\tilde{\psi}(\tau)&=-\frac{Q^{-1}\sqrt{\bar{v}}kQ^{\top^{-1}}}{2}+\frac{M^\top\left(Q^\top Q\right)^{-1}+\left(Q^\top Q\right)^{-1}M}{4}\label{CM_psi},
%\end{align}
%\begin{align*}
%%k\left(-\sqrt{\bar{v}}\coth\left(\sqrt{\bar{v}}t\right)\right)-\bar{w}k=\sqrt{\bar{v}}+\bar{w}\coth\left(\sqrt{\bar{v}}t\right)
%k&=-\left(\sqrt{\bar{v}}\cosh(\sqrt{\bar{v}}\tau)+\bar{w}\sinh(\sqrt{\bar{v}}\tau)\right)^{-1}\left(\sqrt{\bar{v}}\sinh(\sqrt{\bar{v}}\tau)+\bar{w}\cosh(\sqrt{\bar{v}}\tau)\right),\\
%\bar{v}&=Q\left(2B+M^\top Q^{-1}Q^{\top^{-1}}M\right)Q^\top,\\
%\bar{w}&=Q\left(-\frac{M^\top\left(Q^\top Q\right)^{-1}+\left(Q^\top Q\right)^{-1}M}{2}\right)Q^\top.
%\end{align*}
%
%\item \textbf{Variation of constant approach (See \cite{article_gou03})}
\begin{align}
\tilde{\psi}(\tau)&=\psi^\prime+e^{\left(M^{\top}-2\psi^\prime Q^{\top}Q\right)\tau}\Bigg[(-\psi^\prime)^{-1}\Bigg.\nonumber\\
&\Bigg.+2\int_{0}^{\tau}{e^{\left(M-2Q^{\top}Q\psi^\prime\right)s}Q^{\top}Q e^{\left(M^{\top}-2\psi^\prime Q^{\top}Q\right)s}ds}\Bigg]^{-1}e^{\left(M-2Q^{\top}Q\psi^\prime\right)\tau},\\
\tilde{\phi}(\tau)&=Tr\left[\alpha Q^{\top}Q\int_{0}^{\tau}{\tilde{\psi}(s)ds}\right]+a\tau.
\end{align}
%
%\item \textbf{Linearization approach (See \cite{article_Grasselli})}
%\begin{equation}
%\tilde{\psi}(\tau)=D(\tau)^{-1}E(\tau),
%\end{equation}
%with
%\begin{equation}
%\left(\begin{array}{rr}
%E(\tau) & D(\tau)
%\end{array}\right)
%=\left(\begin{array}{rr}
%E(0) & D(0)
%\end{array}\right)exp\left\{\tau
%\left( 
%\begin{array}{rr}
%M & 2Q^{\top}Q\\
%B & -M^{\top}
%\end{array}
%\right)\right\}.
%\end{equation} 
%\begin{equation}
%\tilde{\phi}(\tau)=\frac{\alpha}{2}Tr\left[log\left(D(\tau)\right)+M^{\top}\tau\right]+a\tau.
%\end{equation}
%\end{itemize}
\end{proposition}
\begin{proof}See \cite{article_gou03}.
\end{proof}

Alternative approaches for the pricing of the zero coupon bond may be found in \cite{article_Grasselli} and  \cite{GnoGra2011}.

By looking at the variation of constant approach we can prove the following claim, which constitutes the generalization to the present setting of Corollary 3.5 in \cite{article_KRS}.

\begin{corollary}\label{cor_infty}
\begin{align}
\lim_{\tau\to\infty}{\tilde{\psi}(\tau)}=\psi^\prime.
\end{align}
\end{corollary}
\begin{proof}
Since $\lambda\left(M-2Q^{\top}Q\psi^\prime\right)<0$ by assumption, we know that the integral in the solution for $\tilde{\psi}$ is convergent, moreover, from Lemma \ref{lemma1}, we know that $e^{\left(M-2Q^{\top}Q\psi^\prime\right)\tau}\searrow 0$ as $\tau\to\infty$, hence the proof is complete.
\end{proof}
This last corollary tells us that the function $\tilde{\psi}$ tends to a stability point of the Riccati ODE. This allows us to claim that, as $\tau\to\infty$, we have $\mathcal{R}\left(\tilde{\psi}(\tau)\right)\searrow 0$.

\section{Yield Curve Shapes}
In this section we perform an investigation on the shapes of the yield curve produced by the Wishart short rate model. We will derive a set of sufficient conditions ensuring that certain shapes are attained. We will work with the general diffusion model and show how to replicate, normal, inverse or humped curves. In the appendix we will repeat the same analysis in a simpler version of the model where there will be further limitations on the possible shapes one can obtain. We use the standard dotted notation to represent derivatives w.r.t. time dimensions.

\subsection{Monotonicity of $\tilde{\psi}(\tau)$ and asymptotic behavior of the Yield Curve}
Here we report a result concerning the monotonicity of the function $\tilde{\psi}(\tau)$, which may be found in \cite{article_BCT}. First we recall a result from control theory (see  \cite{book_Brockett70}).

\begin{proposition} (Matrix variation of constants formula) If $\mathbf{\Phi}_{1}(t,t_0)$ is the transition matrix of $\dot{x}(t)=A_{1}(t)x(t)$ and $\mathbf{\Phi}_{2}(t,t_0)$ is the transition matrix for $\dot{x}(t)=A_{2}^{\top}(t)x(t)$, then the solution of
\begin{equation}
\dot{X}(t)=A_{1}(t)X(t)+X(t)A_{2}(t)+F(t),
\end{equation}
with the initial state vector $X(t_0)$, is given by:
\begin{equation}
X(t)=\mathbf{\Phi}_{1}(t,t_0)X(t_0)\mathbf{\Phi}^{\top}_{2}(t,t_0)+\int_{0}^{t}{\mathbf{\Phi}_{1}(t,t_0)F(s)\mathbf{\Phi}^{\top}_{2}(t,t_0)}ds.\label{varConst}
\end{equation}
\end{proposition}

\begin{proposition}
Let $X=(X_s)_{t\leq s \leq T}$ be the stochastic process defined by the dynamics \eqref{genWis}. Then $\tilde{\psi}(\tau)$ is monotonically increasing in $\tau$, i.e., for $\tau_{2}\geq\tau_1$, we have that $\tilde{\psi}(\tau_2)\succeq\tilde{\psi}(\tau_1)$. 
\end{proposition}
\begin{proof}
First, we differentiate \eqref{ODE2}, so as to obtain the following:
\begin{equation}
\ddot{\tilde{\psi}}(\tau)=\dot{\tilde{\psi}}(\tau)M+M^{\top}\dot{\tilde{\psi}}(\tau)-2\dot{\tilde{\psi}}Q^\top Q(\tau)\tilde{\psi}(\tau)-2\tilde{\psi}(\tau)Q^\top Q\dot{\tilde{\psi}}(\tau).
\end{equation}
Next we define $V(\tau)=M-2Q^{\top}Q\tilde{\psi}(\tau)$. Then we may write:
\begin{equation}
\ddot{\tilde{\psi}}(\tau)=\dot{\tilde{\psi}}(\tau)V(\tau)+V^{\top}(\tau)\dot{\tilde{\psi}}(\tau),
\end{equation}
which is solved by:
\begin{equation}
\dot{\tilde{\psi}}(\tau)=\mathbf{\Phi}(\tau,0)\dot{\tilde{\psi}}(0)\mathbf{\Phi}^{\top}(\tau,0),
\end{equation}
for a state transition matrix $\mathbf{\Phi}(\tau,0)$ of the system matrix $V(\tau)$, solving $\dot{\mathbf{\Phi}}(\tau,0)=V^{\top}(\tau)\mathbf{\Phi}(\tau,0)$, $\mathbf{\Phi}(0,0)=I_{d}$. Substitution of the initial condition $\dot{\tilde{\psi}}(0)=B$ yields (see Remark \ref{mathcal_R} and recall from \eqref{ODE_bond_psi} that we have $\dot{\tilde{\psi}}(\tau)=\mathcal{R}(\tilde{\psi}(\tau))$):
\begin{equation}
\dot{\tilde{\psi}}(\tau)=\mathbf{\Phi}(\tau,0)B\mathbf{\Phi}^{\top}(\tau,0).\label{psi_tilde_dot}
\end{equation}
This derivative is positive semidefinite for $B\in S_{d}^{+}$, upon integration on the interval $\left(\tau_1,\tau_2\right)$ we obtain:
\begin{equation}
\tilde{\psi}(\tau_2)-\tilde{\psi}(\tau_1)=\int_{\tau_1}^{\tau_2}{\dot{\tilde{\psi}}(u)du}.
\end{equation}
Therefore $\tilde{\psi}$ is an increasing function of $\tau$.
\end{proof}

As a consequence, we can derive the following useful corollary concerning the function $\mathcal{R}$, defined in Remark \ref{mathcal_R}.
\begin{corollary}\label{cor_2}
$\forall \tau\in\left[0,\infty\right)$ we have:
\begin{align}
\mathcal{R}\left(\tilde{\psi}(\tau)\right)&\in S_d^+\text{ for } B\in S_d^+,\\
\mathcal{R}\left(\tilde{\psi}(\tau)\right)&\in S_d^{++}\text{ for } B\in S_d^{++}.
\end{align}
\end{corollary}
\begin{proof}
%From the initial condition for the bond price we know that $\tilde{\psi}(0)=0$, hence, due to \eqref{ODESimple2}, $\mathcal{R}\left(\tilde{\psi}(0)\right)=B$. Moreover we know that $\tilde{\psi}(\tau)$ is increasing, meaning that the function $\mathcal{R}\left(\tilde{\psi}(\tau)\right)$ is monotonically decreasing in $\tau$. By recalling the limiting behavior of $\tilde{\psi}(\tau)$ in \eqref{limit} we obtain that $\mathcal{R}\left(\tilde{\psi}(\tau)\right)$ tends to $0\in M_{d\times d}$ as $\tau\to\infty$, so we are done.
From equation \eqref{psi_tilde_dot} we know that:
\begin{equation}
\dot{\tilde{\psi}}(\tau)=\mathbf{\Phi}(\tau,0)B\mathbf{\Phi}^{\top}(\tau,0)
\end{equation}
But $\mathcal{R}\left(\tilde{\psi}(\tau)\right)=\dot{\tilde{\psi}}(\tau)$. This shows that $\mathcal{R}\left(\tilde{\psi}(\tau)\right)$ is a congruent transformation of $B$. According to Sylvester's law of inertia the signs of the eigenvalues are unchanged under congruent transformations, hence the claim. 
\end{proof}

We now proceed to show the behavior of the yield curve as $\tau\to0$ and $\tau\to\infty$. First of all we provide the following:

\begin{definition}\label{defi}
The zero coupon yield $Y(\tau,X_t):\mathbb{R}_{\geq0}\times S_{d}^{+}\rightarrow\mathbb{R}_{\geq0}$ is defined as:
\begin{equation}
Y(\tau,X_t)=\frac{\tilde{\phi}(\tau)}{\tau}+\frac{Tr\left[\tilde{\psi}(\tau)X_t\right]}{\tau}.\label{ycdef}
\end{equation}
For fixed $X_t$ we call the function $Y(\tau)=Y(\tau,X_t)$ the \textbf{yield curve}.
\end{definition}
Then we show the following:
\begin{proposition}
Let $X=(X_s)_{t\leq s \leq T}$ be the stochastic process defined by the dynamics \eqref{genWis}. Then the following relations hold true:
\begin{equation}
\lim_{\tau\to0}Y(\tau)=r_t,\label{limZero}
\end{equation}
\begin{equation}
\lim_{\tau\to\infty}Y(\tau)=\mathcal{F}(\psi^\prime),\label{limInfinity}
\end{equation}
with $\mathcal{F}$ as in \eqref{mcF}.
\end{proposition}
\begin{proof}
We start with \eqref{limZero}. By using l'Hospital rule we may write the following:
\begin{equation}
\lim_{\tau\to0}{\frac{\tilde{\phi}(\tau)}{\tau}}=\lim_{\tau\to0}\mathcal{F}(\tilde{\psi}(\tau))=\lim_{\tau\to0}Tr\left[b\tilde{\psi}(\tau)\right]+a=a,
\end{equation}
\begin{align}
\lim_{\tau\to0}{\frac{Tr\left[\tilde{\psi}(\tau)X_t\right]}{\tau}}&=\lim_{\tau\to0}{Tr\left[\mathcal{R}(\tilde{\psi}(\tau))X_t\right]}\nonumber\\
&=\lim_{\tau\to0}{Tr\left[\left(\tilde{\psi}(\tau)M+M^{\top}\tilde{\psi}(\tau)-2\tilde{\psi}(\tau)Q^{\top}Q\tilde{\psi}(\tau)+B\right)X_t\right]}\nonumber\\
&=Tr\left[BX_t\right].
\end{align}
Putting together the two terms we obtain the result since $r_t=a+Tr\left[BX_t\right]$. To show \eqref{limInfinity}, we notice that since $\tilde{\psi}(\tau)\to\psi^\prime$ as $\tau\to\infty$, we have that $\mathcal{R}(\tilde{\psi}(\tau))\to0$. So using again l'Hospital rule and recalling \eqref{mcF} we have:
\begin{equation}
\lim_{\tau\to\infty}{Y(\tau)}=\lim_{\tau\to\infty}{\mathcal{F}(\tilde{\psi}(\tau))}=\mathcal{F}(\psi^\prime),
\end{equation}
so that basically the one dimensional result in Keller-Ressel and Steiner is confirmed also in the present setting.
\end{proof}

\subsection{Yield Curve Shapes in the General Diffusion Model}
In this section we present a set of sufficient conditions, ensuring that the Wishart short rate model produces certain yield curve shapes. Due to the more general structure of the state space, arguments inspired by the scalar case allow us to derive only sufficient conditions, ensuring the attainability of certain yield curve shapes. The absence of a total order relation on $S_d^+$ does not allow us to rule out other possibilities (i.e. more complex shapes). In spite of this we believe that this result is interesting since, e.g. in a calibration setting, we are then able to put ex-ante some constraints to ensure that the model reproduces the yield curve shape observed on the market.
\subsubsection{Statement of the main result}
We begin with a definition:
\begin{definition}\label{defshapes}
Let the Yield curve be defined as in Definition \ref{defi}. We say that the yield curve is:
\begin{itemize}
\item \textbf{normal} if $Y$ is a strictly increasing function of $\tau$, 
\item \textbf{inverse} if $Y$ is a strictly decreasing function of $\tau$,
\item \textbf{humped} if $Y$ has a local maximum and no minimum on $\left[0,\infty\right)$.
\end{itemize}
\end{definition}
We will see that for our particular choice of the model, the arguments employed in \cite{article_KRS} may be easily extended, with some adjustments due to the different state space. As in their setting, we report this lemma.
\begin{lemma}\label{lemmino}
A strictly convex or a strictly concave real function on $\mathbb{R}$ intersects an affine function in at most two points. In the case of two intersection points $p_1<p_2$, the convex function lies strictly below the affine function on the interval $\left(p_1,p_2\right)$; if the function is concave it lies strictly above the affine function on $\left(p_1,p_2\right)$.
\end{lemma}

Before we present the main result, we would like to recall some facts concerning the solution of a class of matrix equations (see e.g. \cite{hojo} Chapter 2 and 6):
\begin{theorem}\label{mat_eq}
Let $A\in\mathbb{C}^{n\times n}$, $B\in\mathbb{C}^{m\times m}$, then we have the following:
\begin{itemize}
\item $\forall$ $D\in\mathbb{C}^{m\times n}$ the equation
\begin{align}
XA+BX=D
\end{align}
has a unique solution if and only if $\alpha+\beta\neq0$, $\forall \alpha\in\sigma(A),\beta\in\sigma(B)$.
\item If $\Re(\alpha+\beta)<0$ $\forall \alpha\in\sigma(A),\beta\in\sigma(B)$ then the following improper integral is convergent and solves the equation:
\begin{align}
X=-\int_{0}^{\infty}{e^{Bs}De^{As}ds}.
\end{align}
\end{itemize}
\end{theorem}
If $B=A^{\top}$ and $D=-Q$ for $Q\in S_d^+$, we have the following corollary:
\begin{corollary}
$\Re\left(\lambda\left(A\right)\right)<0\leftrightarrow A^{\top}X+XA=-Q\quad X,Q\in S_d^+$
\end{corollary}
Now, we report our main result on the shapes of the yield curve:

\begin{theorem}\label{BigThYCS}
Consider a short rate model in which the risk neutral dynamics of the short rate is driven by the process $X_t$, defined by the dynamics \eqref{genWis}. Let $B \in S_d^{++}$. Define $M^\star:=M-2Q^\top Q\psi^\prime$ and:
\begin{align}
b_{norm}:=\alpha\int_{0}^{\infty}{e^{M^\star s}Q^{\top}Q e^{M^{\star\top}s}ds},\qquad b_{inv}:=\alpha\int_{0}^{\infty}{e^{Ms}Q^{\top}Q e^{M^{\top}s}ds}.
\end{align}
Then the following holds:
\begin{itemize}
\item The yield curve is normal if $X_t\prec b_{norm}$.
\item The yield curve is inverse if $X_t\succ b_{inv}$.
\item The yield curve is humped if $b_{norm}\prec X_t\prec b_{inv}$.
\end{itemize} 

\end{theorem}

\subsubsection{Proof of Theorem \ref{BigThYCS}}
We define the function $\mathcal{H}(\tau):\mathbb{R}_{\geq 0}\rightarrow\mathbb{R}$ by 
\begin{align}
\mathcal{H}(\tau):=Y(\tau,X_t)\tau=\tilde{\phi}(\tau)+Tr\left[\tilde{\psi}(\tau)X_t\right]\label{calH_def}.
\end{align}
Recalling the system of equations satisfied by $\tilde{\phi}$ and $\tilde{\psi}$, given by equations \eqref{ODE1} and \eqref{ODE2}, we can compute the derivatives of this function. For the first derivative we have:
\begin{align}
\mathcal{H}'(\tau)&=Tr\left[\alpha Q^{\top}Q\tilde{\psi}(\tau)\right]+a\nonumber\\
&+Tr\left[\left(\tilde{\psi}(\tau)M+M^{\top}\tilde{\psi}(\tau)-2\tilde{\psi}(\tau)Q^{\top}Q\tilde{\psi}(\tau)\right)X_t\right],
\end{align}
whereas for the second we have:
\begin{align}
\mathcal{H}''(\tau)&=Tr\left[\alpha Q^{\top}Q\mathcal{R}\left(\tilde{\psi}(\tau)\right)\right.\nonumber\\
&+\left(\mathcal{R}\left(\tilde{\psi}(\tau)\right)M+M^{\top}\mathcal{R}\left(\tilde{\psi}(\tau)\right)-2\mathcal{R}\left(\tilde{\psi}(\tau)\right)Q^{\top}Q\tilde{\psi}(\tau)\right.\nonumber\\
&\left.\left.-2\tilde{\psi}(\tau)Q^{\top}Q\mathcal{R}\left(\tilde{\psi}(\tau)\right)\right)X_t\right].
\end{align}
Now, notice the following:
\begin{align}
Tr\left[M^{\top}\mathcal{R}\left(\tilde{\psi}(\tau)\right)X_t\right]&=Tr\left[\mathcal{R}\left(\tilde{\psi}(\tau)\right)X_tM^{\top}\right],\nonumber\\
Tr\left[-2\mathcal{R}\left(\tilde{\psi}(\tau)\right)Q^{\top}Q\tilde{\psi}(\tau)X_t\right]&=Tr\left[-2\tilde{\psi}(\tau)Q^{\top}Q\mathcal{R}\left(\tilde{\psi}(\tau)\right)X_t\right].\nonumber
\end{align}
The second equality being justified since the matrices involved are symmetric.
This means that we may rewrite $\mathcal{H}''(\tau)$ as follows:
\begin{align}
\mathcal{H}''(\tau)&=Tr\left[\mathcal{R}\left(\tilde{\psi}(\tau)\right)\left(\overbrace{\alpha Q^{\top}Q+MX_t+X_tM^{\top}-4Q^{\top}Q\tilde{\psi}(\tau)X_t}^{=:k(\tau)}\right)\right]\nonumber\\
&=Tr\left[\mathcal{R}\left(\tilde{\psi}(\tau)\right)k(\tau)\right].
\end{align}
Finally, we can equivalently work with a symmetrization of the function $k(\tau)$:
\begin{align}
\mathcal{H}''(\tau)&=Tr\left[\mathcal{R}\left(\tilde{\psi}(\tau)\right)k(\tau)\right]\nonumber\\
&=Tr\left[\mathcal{R}\left(\tilde{\psi}(\tau)\right)\frac{k(\tau)+k^{\top}(\tau)}{2}\right]\\
&=Tr\left[\mathcal{R}\left(\tilde{\psi}(\tau)\right)\left(\overbrace{\alpha Q^{\top}Q+\left(M-2Q^{\top}Q\tilde{\psi}(\tau)\right)X_t+X_t\left(M^{\top}-2\tilde{\psi}(\tau)Q^{\top}Q\right)}^{=:\tilde{k}(\tau)}\right)\right]\nonumber\\
&=Tr\left[\mathcal{R}\left(\tilde{\psi}(\tau)\right)\tilde{k}(\tau)\right].\nonumber
\end{align}
To start the derivation of the sufficient conditions for the shapes of the yield curve we look at $\mathcal{H}''(0)$. A sufficient condition for $\mathcal{H}''(0)=0$ is $\tilde{k}(0)=0$, i.e.:
\begin{align}
MX_t+X_tM^{\top}=-\alpha Q^\top Q.
\end{align}
The solution to this equation is given, according to Theorem \ref{mat_eq}, by:
\begin{equation}
b_{inv}=\alpha\int_{0}^{\infty}{e^{Ms}Q^{\top}Q e^{M^\top s}ds}.
\end{equation}
Recalling that the eigenvalues of $M$ are negative (see Assumption \ref{assu_eigen}) we have that:
\begin{align}
 \text{if }X_t\succ b_{inv}\text{ then }\mathcal{H}''(0)<0\label{Hsecond_1},\\
 \text{if }X_t\prec b_{inv}\text{ then }\mathcal{H}''(0)>0\label{Hsecond_2}.
\end{align}
Recall that $b_{inv}$ above is defined as the solution of the Lyapunov equation resulting from $\tilde{k}(\tau)$ when $\tau=0$, by noting that $\tilde{\psi}(0)=0$.
Then consider the equation:
\begin{align}
\alpha Q^{\top}Q+\left(M-2Q^{\top}Q\tilde{\psi}(\tau)\right)X_t+X_t\left(M^{\top}-2\tilde{\psi}(\tau)Q^{\top}Q\right)=0.
\end{align}
By recalling that $\tilde{\psi}(\tau)\nearrow\psi^\prime$ as $\tau\to\infty$, we have exactly a solution at infinity if $X_t$ solves:
\begin{align}
\alpha Q^{\top}Q+\left(M-2Q^{\top}Q\psi^\prime\right)X_t+X_t\left(M^{\top}-2\psi^\prime Q^{\top}Q\right)=0.
\end{align}
The solution to the equation above is:
\begin{align}
b_{norm}=\alpha\int_{0}^{\infty}{e^{\left(M-2Q^{\top}Q\psi^\prime\right)s}Q^{\top}Q e^{\left(M^\top-2\psi^\prime Q^{\top}Q\right) s}ds}.
\end{align}
In the sequel we will prove that this is the crucial level ensuring the presence of a zero for $\mathcal{H}''(\tau)$. The next lemma establishes an order relation between $b_{inv}$ and $b_{norm}$.

\begin{lemma}\label{lemma11}
$b_{inv}\succ b_{norm}$.
\end{lemma}
\begin{proof}
Consider the function $\tilde{k}(\tau)$. It is easy to see that:
\begin{itemize}
\item if $X_t=b_{inv}$, then $\tilde{k}(\tau)=-2Q^\top Q\tilde{\psi}(\tau)b_{inv}-2b_{inv}\tilde{\psi}(\tau)Q^\top Q $, 
\item if $X_t=b_{norm}$, then
\begin{align}
\tilde{k}(\tau)=0-2Q^{\top}Q\left(\tilde{\psi}(\tau)-\psi^\prime\right)b_{norm}-b_{norm}\left(\tilde{\psi}(\tau)-\psi^\prime\right)Q^{\top}Q2\in S_d^+,
\end{align}
since $\psi^\prime\succ\tilde{\psi}(\tau)\quad\forall\tau$.
\end{itemize}
This shows that $\tilde{k}(\tau)|_{X_t=b_{norm}}\succ\tilde{k}(\tau)|_{X_t=b_{inv}}$. Notice now that, by expliciting the dependence of $\tilde{k}(\tau)$ on $X_t$ we have
\begin{align}
\tilde{k}(\tau,X_t)=\alpha Q^{\top}Q+\left(M-2Q^{\top}Q\tilde{\psi}(\tau)\right)X_t+X_t\left(M^{\top}-2\tilde{\psi}(\tau)Q^{\top}Q\right),
\end{align}
where $M-2Q^{\top}Q\tilde{\psi}(\tau)$ has negative eigenvalues by Assumption \ref{assu_eigen}. It thus follows that $b_{inv}\succ b_{norm}$. 
\end{proof}
Let us now prove the following:
\begin{lemma}\label{lemma12}
If $X_t\succ b_{inv}$, then $\mathcal{H}''(\tau)<0$, $\forall\tau\in\left(0,\infty\right)$.
\end{lemma}
\begin{proof}
Let $X_t=b_{inv}+C$, $C\in S_d^+$, then:
\begin{align}
\tilde{k}(\tau)&=-2Q^\top Q\tilde{\psi}(\tau)b_{inv}-b_{inv}\tilde{\psi}(\tau)Q^\top Q2\nonumber\\
&+\left(M-2Q^\top Q\tilde{\psi}(\tau)\right)C+C\left(M^\top-\tilde{\psi}(\tau)Q^\top Q2\right)
\end{align}
This means that $\tilde{k}$ is symmetric with negative eigenvalues (see Assumption \ref{assu_eigen}), $\forall \tau \in \left[0,\infty\right)$, but then, being $\mathcal{R}(\tilde{\psi}(\tau))\in S_d^{++}$, it follows that $\mathcal{H}''(\tau)<0$.
\end{proof}

Next, we prove an existence result for $\tau^\star\in(0,\infty)$ s.t. $\mathcal{H}''(\tau^\star)=0$. 
\begin{lemma}\label{lemma13}
If $b_{inv}\succ X_t \succ b_{norm}$, then $\exists\tau^\star\in\left(0,\infty\right)$ s.t. $\mathcal{H}''(\tau^\star)=0$.
\end{lemma}
\begin{proof}
Let $X_t=b_{norm}+C$, $C\in S_d^+$ s.t. $X_t\prec b_{inv}$. Then:
\begin{align}
\tilde{k}(\tau)&=\alpha Q^{\top}Q+\left(M-2Q^{\top}Q\tilde{\psi}(\tau)\right)\left(b_{norm}+C\right)\nonumber\\
&+\left(b_{norm}+C\right)\left(M^{\top}-2\tilde{\psi}(\tau)Q^{\top}Q\right).
\end{align}
But this allows us to claim that:
\begin{align}
\lim_{\tau\to\infty}\tilde{k}(\tau)=\left(M-2Q^{\top}Q\psi^\prime\right)C+C\left(M^{\top}-2\psi^\prime Q^{\top}Q\right).
\end{align}
This matrix is symmetric with negative eigenvalues. Recall that since $b_{inv}\succ X_t \succ b_{norm}$, we have that $\tilde{k}(0)\in S_d^{++}$. Then we observe that $\tilde{k}(\tau)$ is a continuous function of $\tau$, meaning that there must exist a $\tau^\prime$ s.t., for $\tau>\tau^\prime$, the eigenvalues of $\tilde{k}(\tau)$ are negative. We can now look at $\mathcal{H}''(\tau)$. We recall that:
\begin{align}
\mathcal{H}''(\tau)=Tr\left[\mathcal{R}\left(\tilde{\psi}(\tau)\right)\tilde{k}(\tau)\right].\nonumber
\end{align}
From Corollary \ref{cor_2} we have $\mathcal{R}\left(\tilde{\psi}(\tau)\right)\in S_d^{++}$. We notice also that $\mathcal{H}''$ is a continuous function of $\tau$. Furthermore, we have $\mathcal{H}''(0)>0$ because $b_{inv}\succ X_t$, see \eqref{Hsecond_2}. From the previous discussion regarding  $\tilde{k}(\tau)$ we have that, for  $\tau>\tau^\prime$, the second derivative is of the form:
\begin{align}
\mathcal{H}''(\tau)=Tr\left[\mathcal{R}\left(\tilde{\psi}(\tau)\right)\left(-\mathcal{K}\right)\right],
\end{align}
where $\mathcal{K}$ is a symmetric matrix with negative eigenvalues. This means that the second derivative will be negative. By recalling the positiveness of the starting value and the continuity property w.r.t. $\tau$, thanks to the mean value theorem, we can argue that there must exist a $\tau^\star$ s.t. $\mathcal{H}''(\tau)=0$ as we wanted.
\end{proof}
Along the same lines we can prove the following:
\begin{lemma}\label{lemma14}
If $X_t \prec b_{norm}$, then there exists no $\tau^\star \in \left(0,\infty\right)$ s.t. $\mathcal{H}''(\tau^\star)=0$.
\end{lemma}
We finally prove a a result allowing us to conclude that the zero of $\mathcal{H}''(\tau)$ is unique.
\begin{lemma}\label{lemma15}
$\tilde{k}(\tau)$ is monotonically decreasing i.e., for $\tau_2>\tau_1$, we have $\tilde{k}(\tau_2)-\tilde{k}(\tau_1)\in S_d^-$.
\end{lemma}
\begin{proof}
Differentiate $\tilde{k}(\tau)$, so as to obtain:
\begin{align}
\dot{\tilde{k}}(\tau)=-2Q^{\top}Q\mathcal{R}\left(\tilde{\psi}(\tau)\right) X_t-2X_t\mathcal{R}\left(\tilde{\psi}(\tau)\right)Q^{\top}Q.
\end{align}
Then we can write:
\begin{align}
\tilde{k}(\tau_2)-\tilde{k}(\tau_1)=-2\int_{\tau_1}^{\tau_2}{\left(Q^{\top}Q\mathcal{R}\left(\tilde{\psi}(s)\right) X_t+X_t\mathcal{R}\left(\tilde{\psi}(s)\right)Q^{\top}Q \right)ds}.
\end{align}
Since the RHS is a symmetric matrix with negative eigenvalues, we have the claim.
\end{proof}
Now, since $\tilde{k}(\tau)$ is monotonically decreasing, we obtain that if there exists a value for $\tau$ s.t. $\tilde{k}(\tau)=0$, then this point in time must be unique.

By relying on Lemmas \ref{lemma11}, \ref{lemma12}, \ref{lemma13}, \ref{lemma14} and \ref{lemma15} and the condition on the sign of $\mathcal{H}''(0)$ in \eqref{Hsecond_2} and \eqref{Hsecond_1}, we can argue the following:
\begin{itemize}
\item if $X_t\prec b_{norm}$, then $\mathcal{H}$ is strictly convex on $\left(0,\infty\right)$.
\item if $b_{norm}\prec X_t\prec b_{inv}$, then $\mathcal{H}$ is strictly convex on $\left(0, \tau^\star\right)$ and strictly concave on $\left(\tau^\star,\infty\right)$.
\item if $X_t\succ b_{inv}$, then $\mathcal{H}$ is strictly concave on $\left(0,\infty\right)$.
\end{itemize}

We use these findings on the convexity of $\mathcal{H}$ to determine our conclusions on the convexity of the yield curve. We consider the equation
\begin{align}
\mathcal{H}(\tau)=c\tau, \qquad\tau\in\left[0,\infty\right)\label{Hctau2},
\end{align}
for some fixed $c\in\mathbb{R}$. Since $\mathcal{H}(0)=0$ (see \eqref{calH_def}), this equation has at least one solution, i.e. $\tau_0=0$. 

Now if $X_t\succ b_{inv}$, then $\mathcal{H}$ is strictly concave on $\left[0,\infty\right)$, and according to Lemma \ref{lemmino},  equation \eqref{Hctau2} has at most one additional solution, $\tau_1$. When the solution exists, $\mathcal{H}(\tau)$ crosses $c\tau$ from above at $\tau_1$.

If $X_t\prec b_{norm}$, then $\mathcal{H}(\tau)$ is strictly convex on $\left[0,\infty\right)$ and again has at most one additional solution $\tau_2$. If the solution exists, $\mathcal{H}(\tau)$ crosses $c\tau$ from below at $\tau_2$. 

In the final case, i.e. if $b_{norm}\prec X_t\prec b_{inv}$, there exists a $\tau^{\star}$, the zero of $\mathcal{H}''(\tau)$, such that $\mathcal{H}(\tau)$ is strictly convex on $\left(0,\tau^{\star}\right)$ and strictly concave on $\left(\tau^{\star},\infty\right)$. This implies that there can exist at most two additional solutions $\tau_1,\tau_2$ to \eqref{Hctau2}, with $\tau_1<\tau^{\star}<\tau_2$, such that $c\tau$ is crossed from below at $\tau_1$ and from above at $\tau_2$. By definition, every solution to \eqref{Hctau2}, $\tau_0=0$ excluded, is also a solution to:
\begin{align}
Y\left(\tau,X_t\right)=c,\qquad\tau\in\left(0,\infty\right),\label{Yc2}
\end{align}
with $X_t$ fixed. The properties of crossing from above/below are preserved since $\tau$ is positive. This means that:
\begin{itemize}
\item if $X_t\succ b_{inv}$, then \eqref{Yc2} has at most a single solution, i.e. every horizontal line is crossed by the yield curve in at most a single point. If it is crossed, it is crossed from above, hence we conclude that $Y\left(\tau,X_t\right)$ is a strictly decreasing function of $\tau$, meaning that the yield curve is inverse.
\item $X_t\prec b_{norm}$, then again \eqref{Yc2} has at most a single solution. If the solution is crossed, it is crossed from below, hence we conclude that $Y\left(\tau,X_t\right)$ is a strictly increasing function of $\tau$, meaning that the yield curve is normal.
\item In the last case, i.e. if $b_{norm}\prec X_t\prec b_{inv}$, then we have at most two additional solutions. If they are crossed, the first is crossed from below and the second from above. This allows us to conclude that the yield curve is humped.
\end{itemize}
In Figure \eqref{fig:ycshapes} we plot a visualization of the results of the theorem.
\newpage
\begin{figure}[H]
\centering
\subfloat{\label{fig:normal2}\includegraphics[scale=0.24]{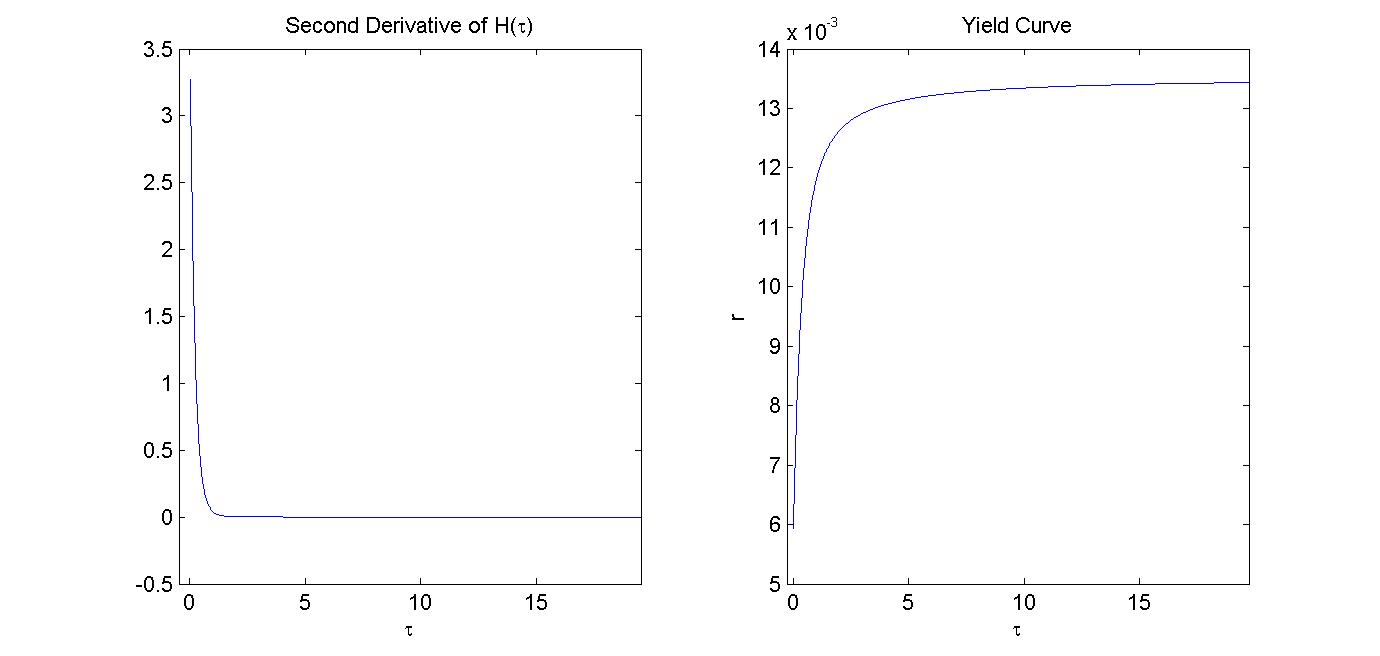}}\\
\subfloat{\label{fig:hump2}\includegraphics[scale=0.22]{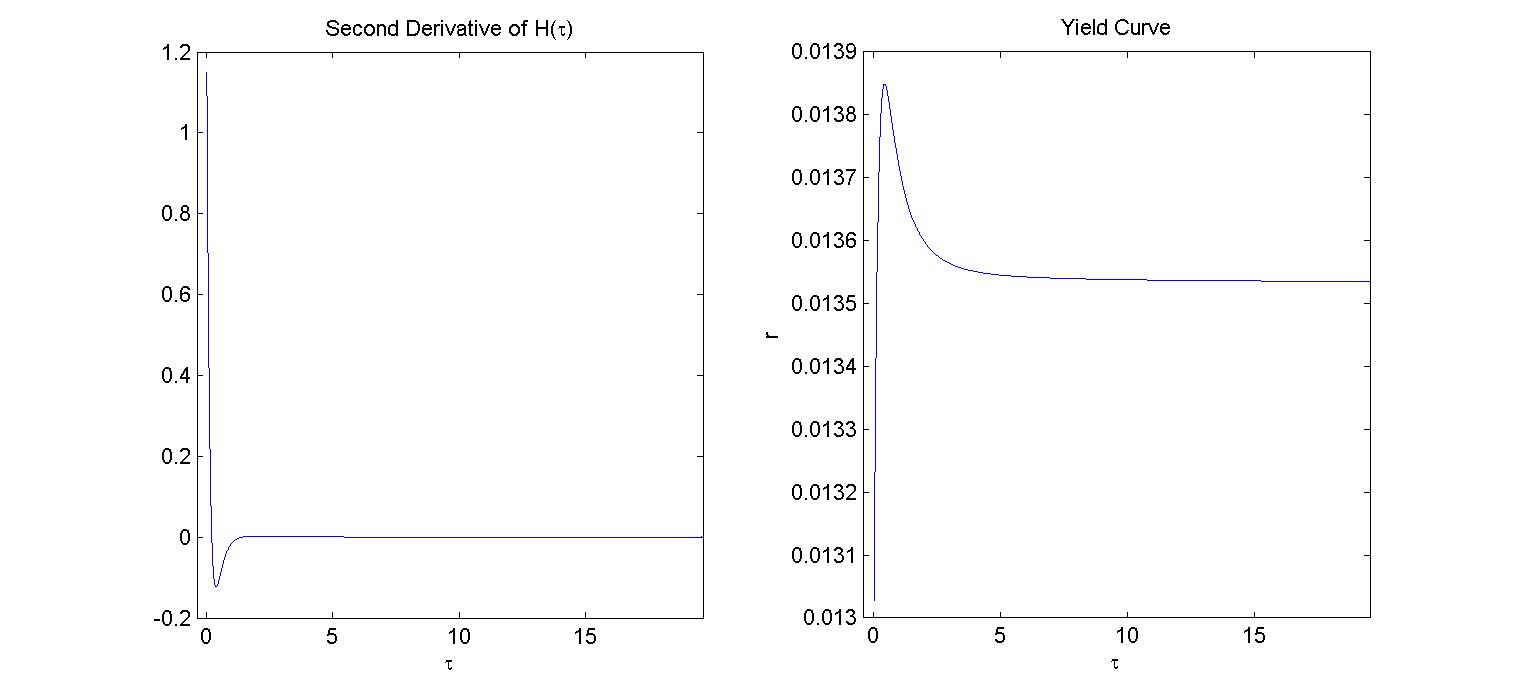}}\\
\subfloat{\label{fig:inverse2}\includegraphics[scale=0.25]{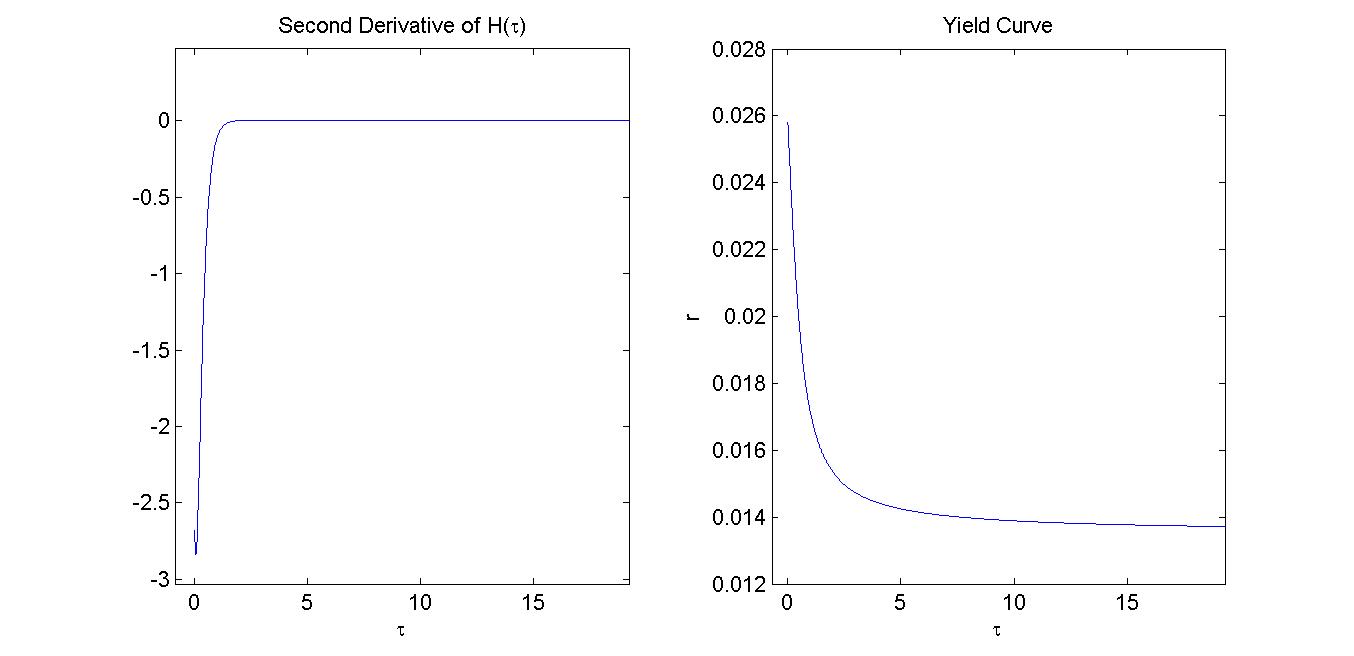}}
\caption{Yield curve shapes for different values of $X_t$}
  \label{fig:ycshapes}
\end{figure}

\section{Discussion on the parameters}
In this section we provide some intuition on the impacts of the parameters on the shape of the yield curve. As a starting point we use the following values for the parameters of the model:
\begin{align}
X_0&=\left(\begin{array}{cc} 0.21&0.003\\ 0.003&0.7\end{array}\right),\quad M=\left(\begin{array}{cc} -1.4&0.1\\0.1&-1.3\end{array}\right),\nonumber\\
Q&=\left(\begin{array}{cc} 1&0.2\\0.3&0.5\end{array}\right),\quad \alpha=3.1,\nonumber\\
B&=\left(\begin{array}{cc} 0.01&0.005\\0.005&0.02\end{array}\right).
\end{align}
With these values, a numerical implementation shows that the yield curve is normal. In the following experiments we will perturbate single elements of the matrices and look at the impact on the yield curve. 
\begin{figure}[H]
  \centering
  \subfloat[$M_{11}$]{\label{fig:M11}\includegraphics[scale=0.04]{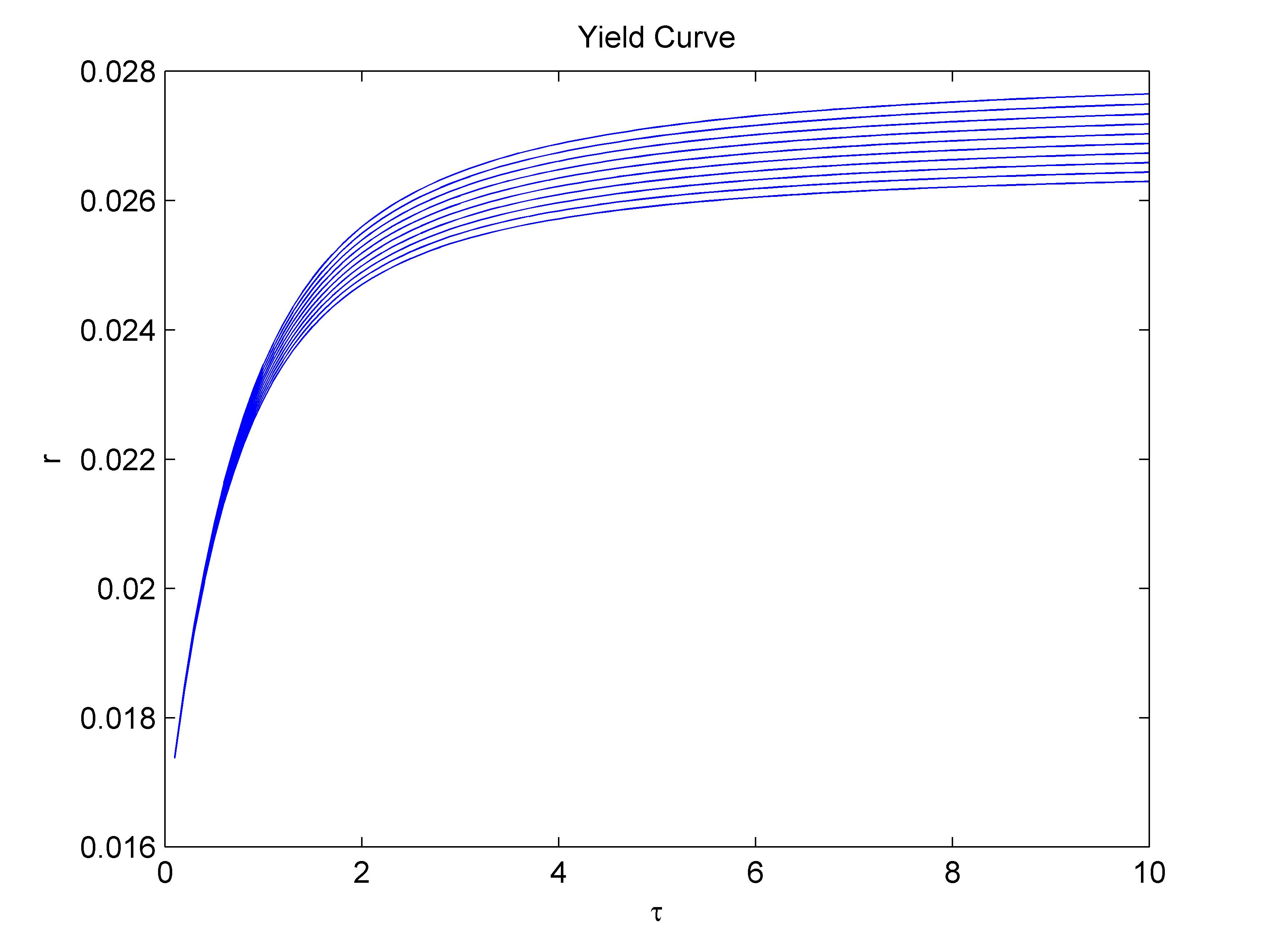}}
 \subfloat[$M_{12}$]{\label{fig:M12}\includegraphics[scale=0.04]{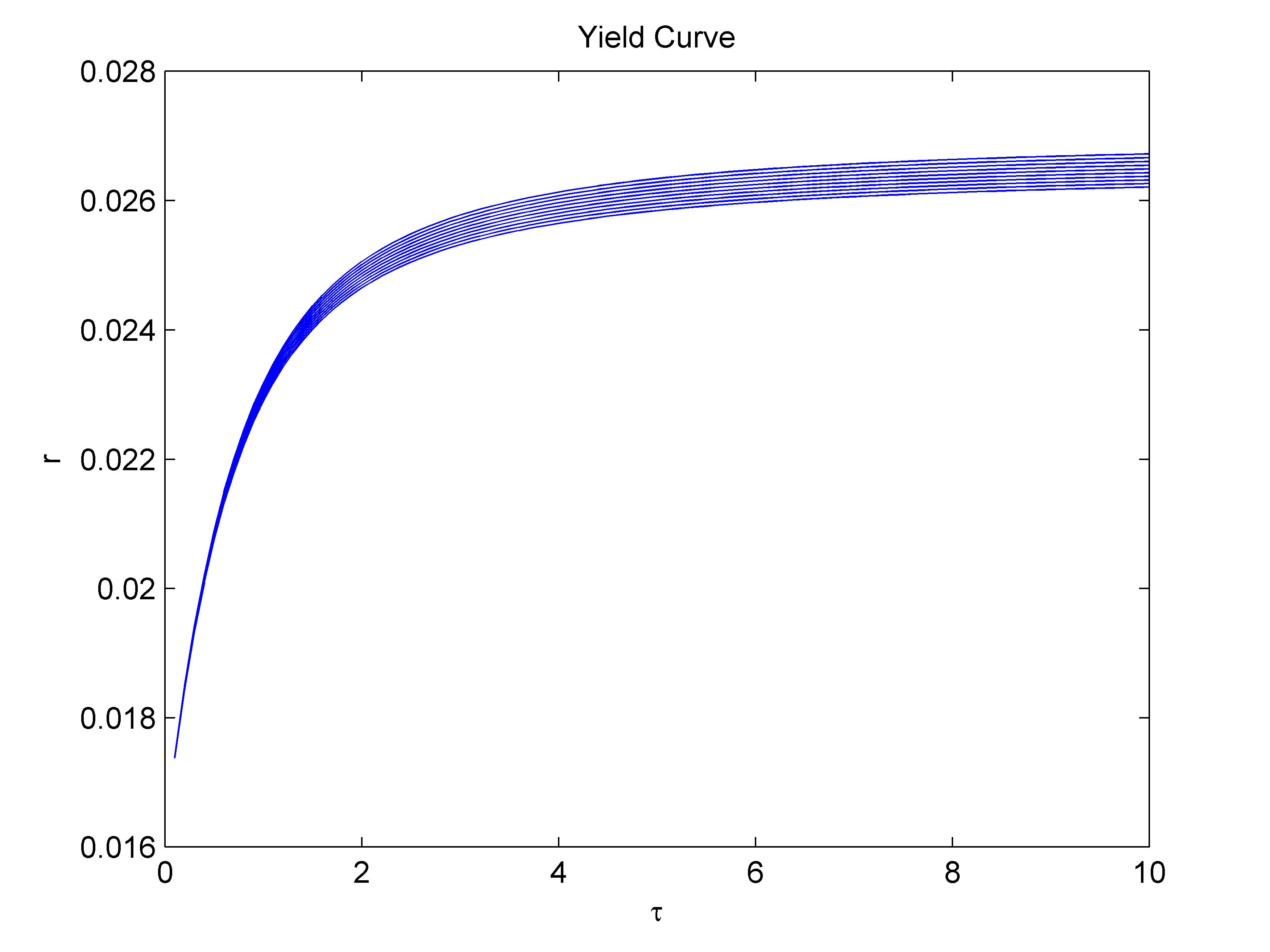}}\\
 \subfloat[$M_{21}$]{\label{fig:M21}\includegraphics[scale=0.04]{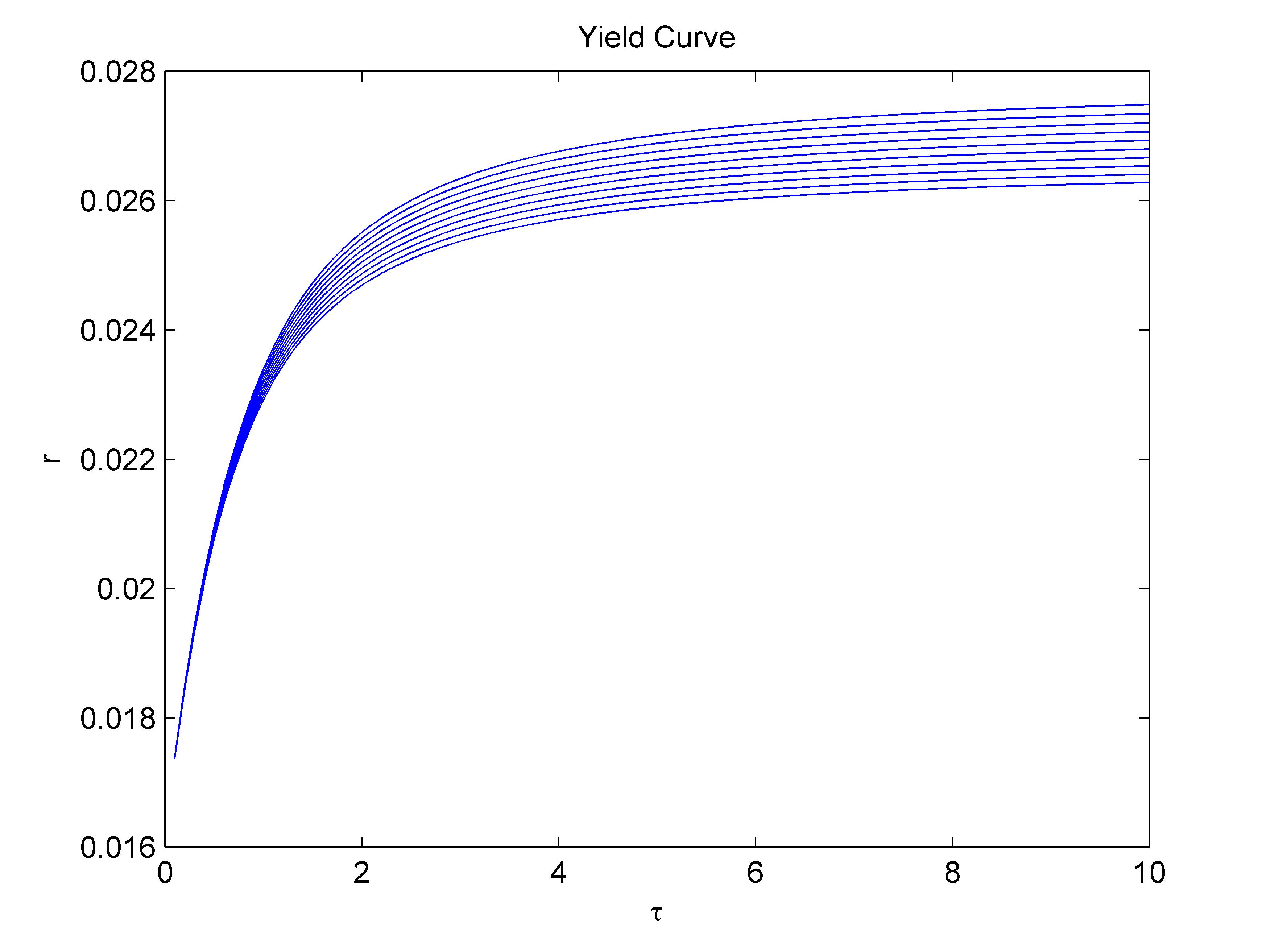}}
 \subfloat[$M_{22}$]{\label{fig:M22}\includegraphics[scale=0.04]{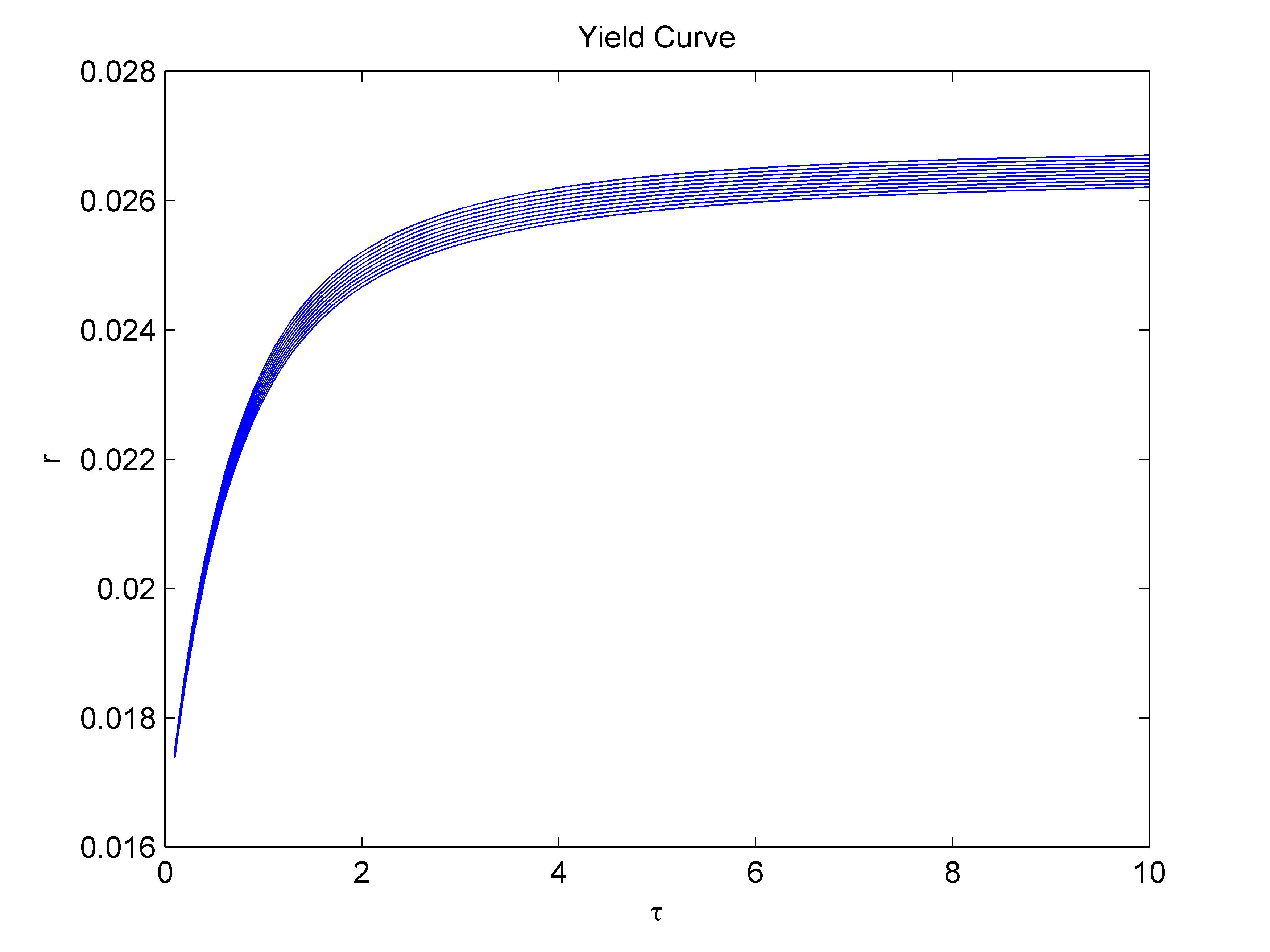}}
  \caption{In this figure we show the effect of a perturbation of the single elements of the matrix $M$. We use a sequence of numbers $\eta=0.01:0.01:0.1$ and add $\eta_i$ to one of the elements of $M$ while leaving the other elements unchanged. When we add $\eta_i$ to the elements on the main diagonal the yield curve is shifted upwards. The same happens with off-diagonal elements.}
  \label{fig:M}
\end{figure}
In Figure \eqref{fig:M}, we perturbate the matrix $M$, by introducing a sequence $\eta=0.01:0.01:0.1$. and we add the values of $\eta_i$ to the single entries of the matrix. It turns out that in all cases we have an upward shift of the yield curve.
\begin{figure}[H]
  \centering
  \subfloat[$Q_{11}$]{\label{fig:Q11}\includegraphics[scale=0.04]{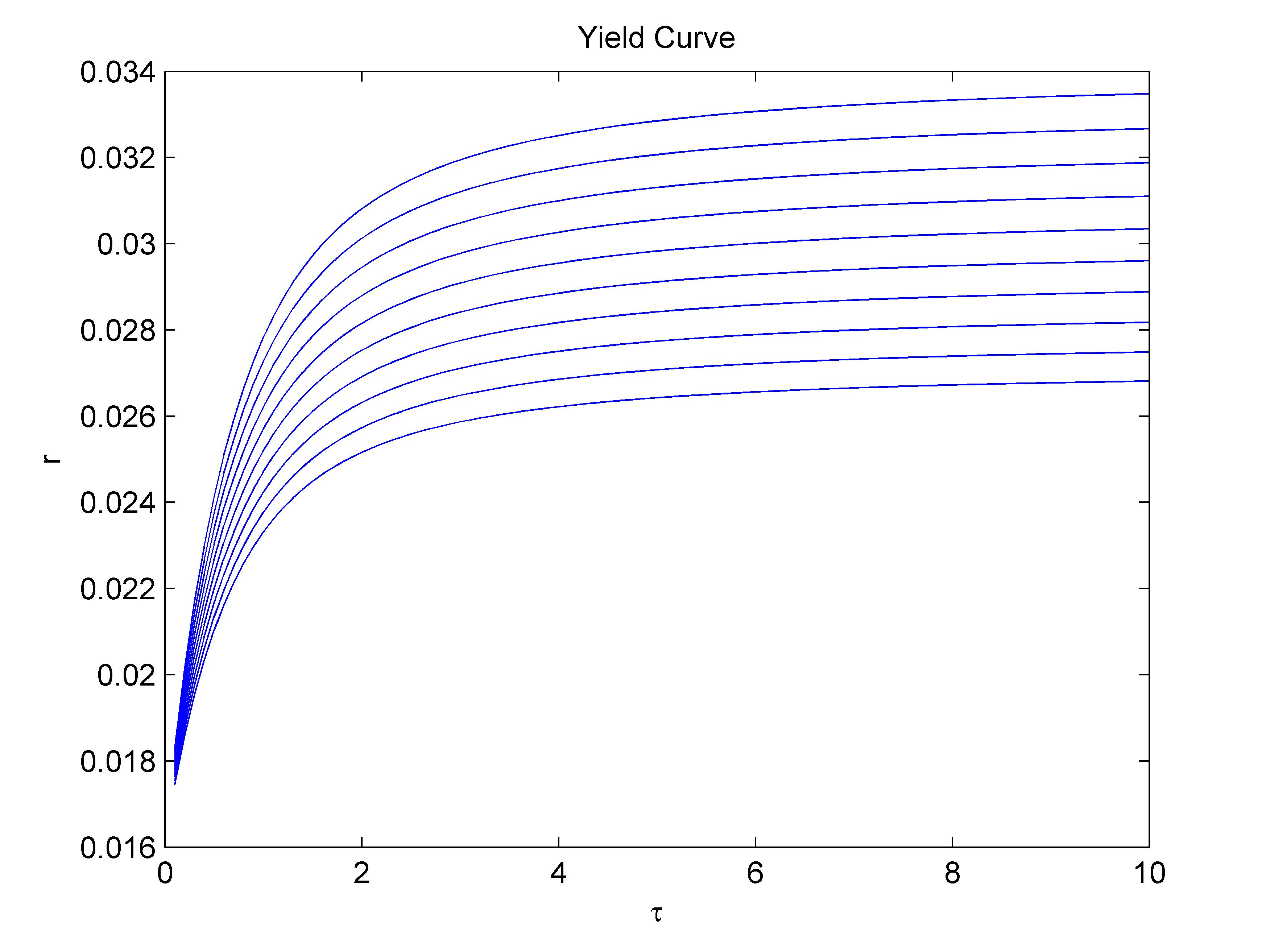}}
 \subfloat[$Q_{12}$]{\label{fig:Q12}\includegraphics[scale=0.04]{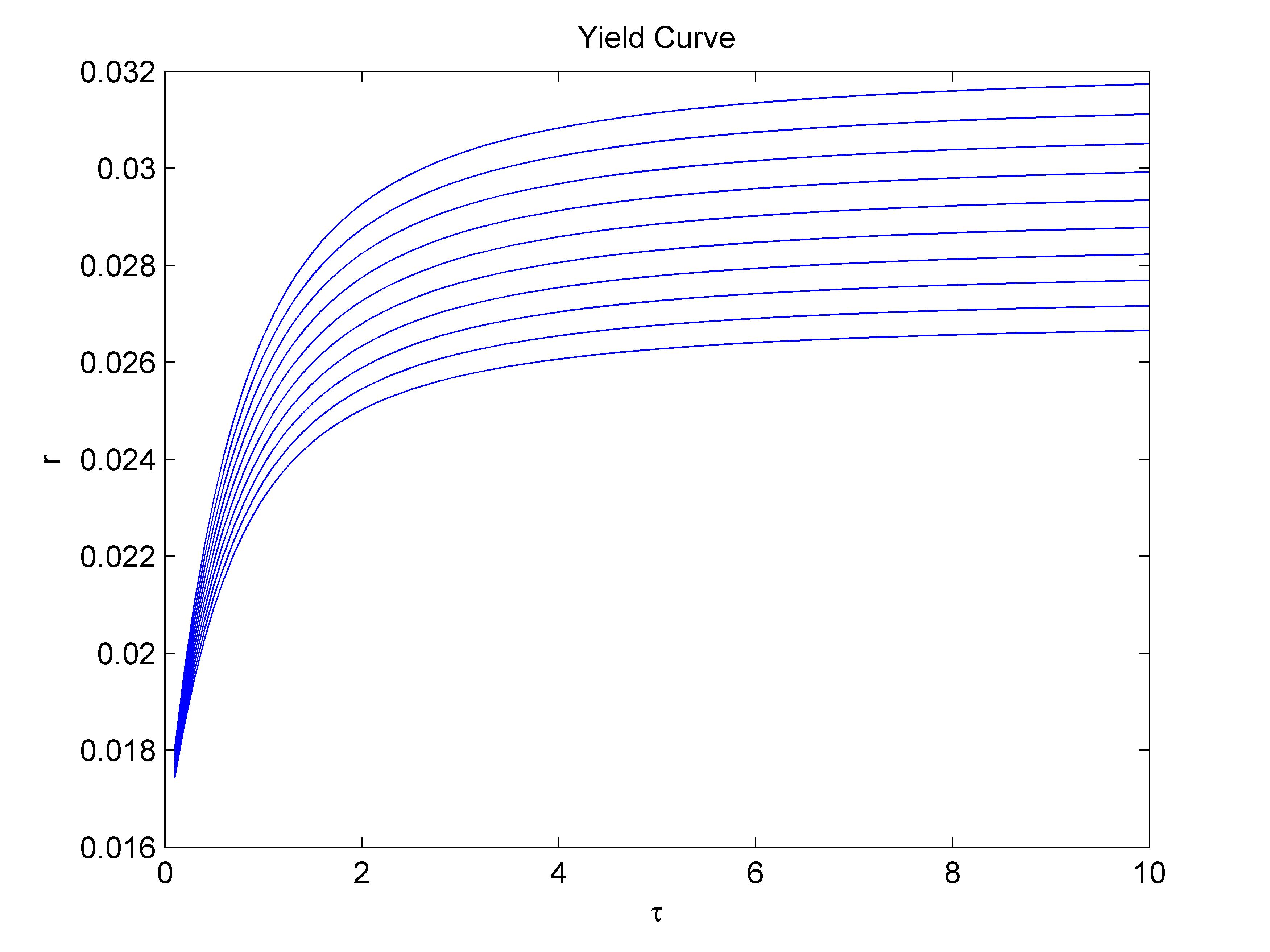}}\\
 \subfloat[$Q_{21}$]{\label{fig:Q21}\includegraphics[scale=0.04]{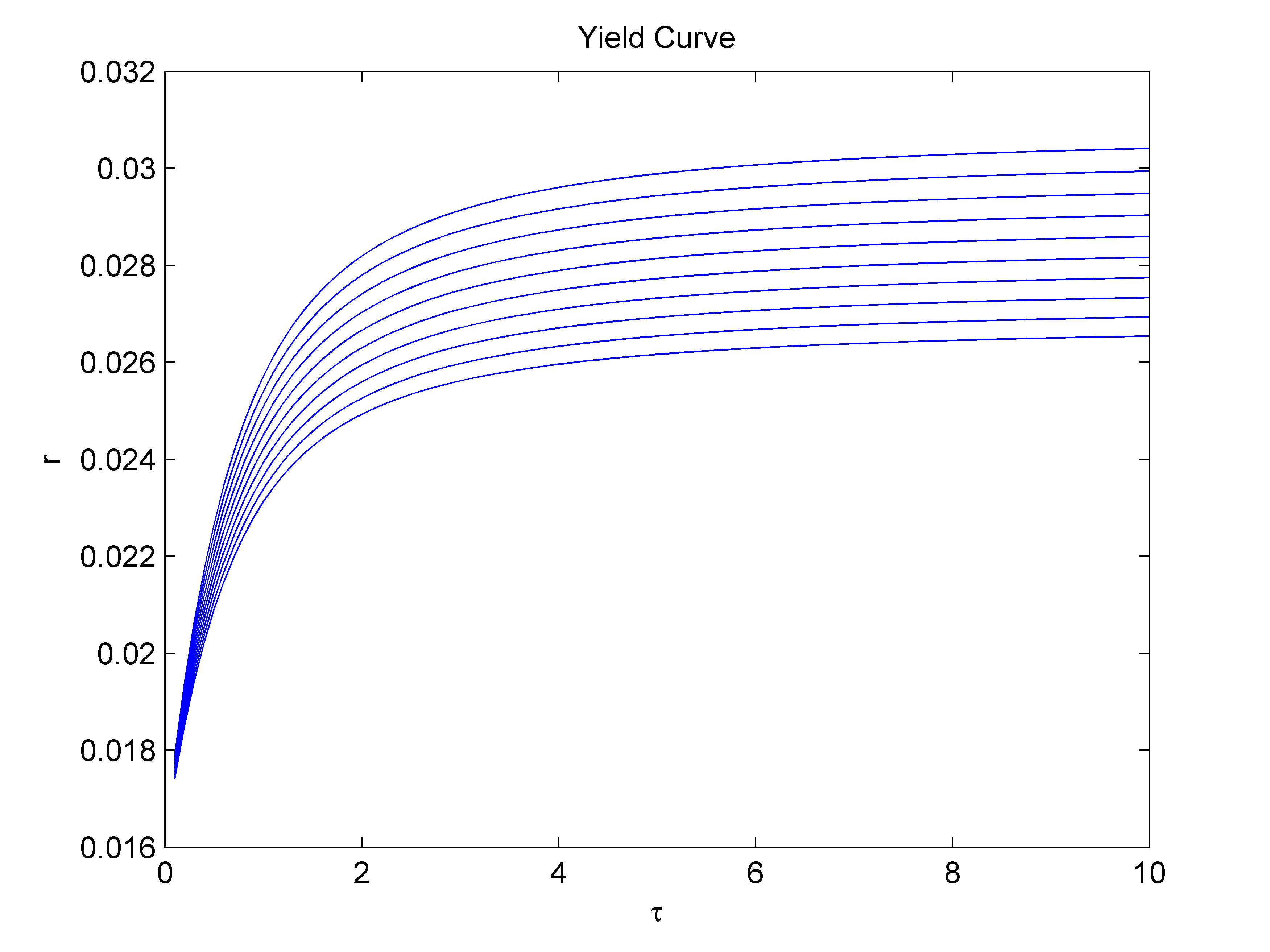}}
 \subfloat[$Q_{22}$]{\label{fig:Q22}\includegraphics[scale=0.04]{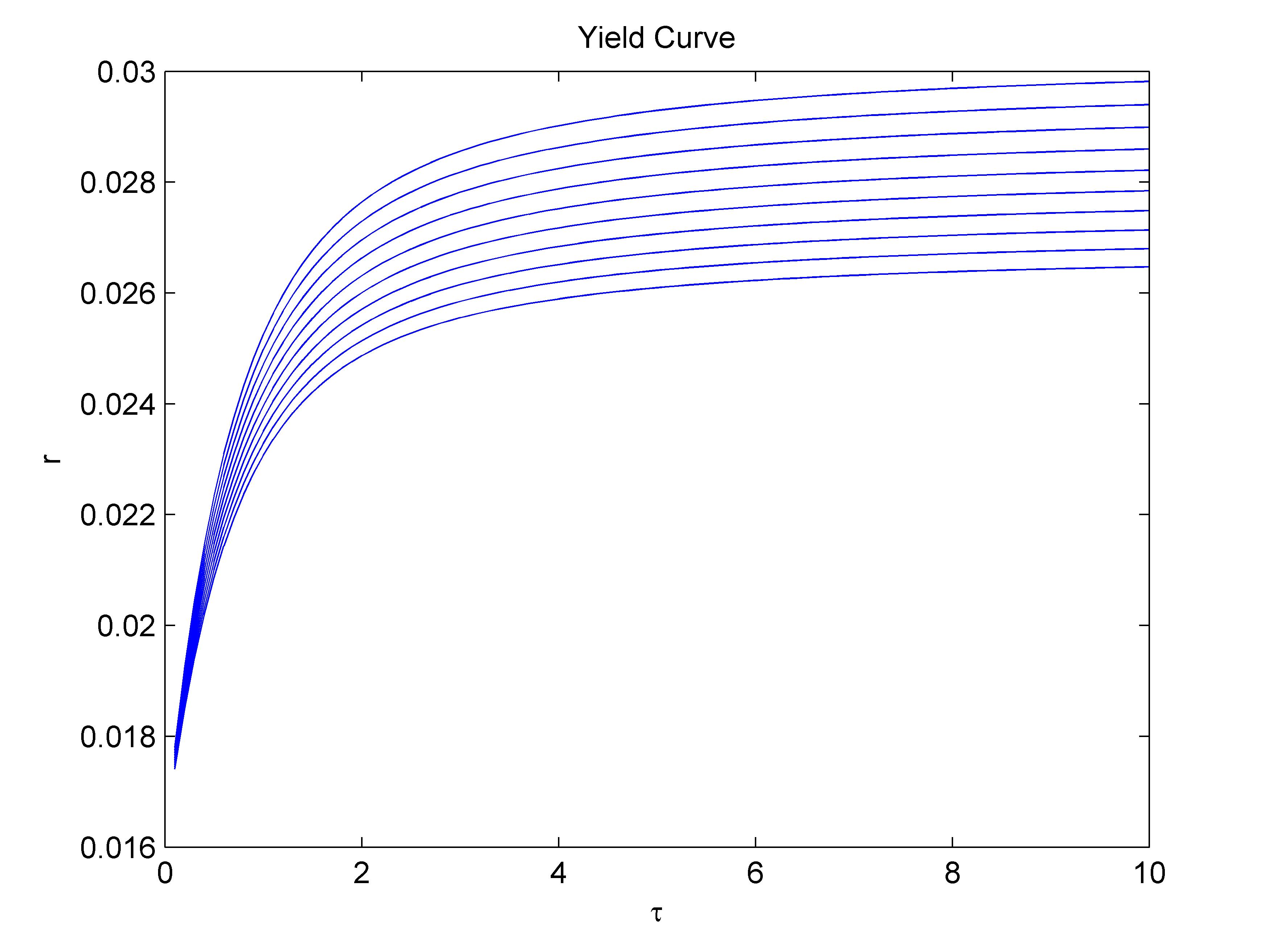}}
  \caption{In this figure we show the effect of a perturbation of the single elements of the matrix $Q$. We use a sequence of numbers $\eta=0.01:0.01:0.1$ and add $\eta_i$ to one of the elements of $Q$ while leaving the other elements unchanged. When we add $\eta_i$ to the elements on the main diagonal the yield curve is shifted upwards. The same happens with off-diagonal elements.}
  \label{fig:Q}
\end{figure}
The impact that we obtain when we perturbate the elements of the matrix $Q$ is similar: we obtain again an upward shift. We would like to perform a more interesting experiment. To this end, we will consider now a larger (w.r.t. the partial order relation on $S_d^+$) value for $X_0$, more precisely:
\begin{align}
X_0&=\left(\begin{array}{cc} 0.3780&0.0054\\0.0054&1.2600\end{array}\right).
\end{align}
With this starting value, we have that the yield curve is humped. We perform the same perturbations as before and notice that this time the impact is more varied. When we look at the impact of perturbations of $M$ we notice shapes ranging from nearly normal ($M_{11}$ and $M_{21}$) till humped shapes ($M_{12}$ and $M_{22}$). Anyhow we notice that the impact on the shape of the yield curve is quite strong.

Finally, we work with $Q$. It turns out that the effect of the diffusion matrix is very relevant. Recall that with our starting value for the factor process, we have a humped curve. Figure \eqref{fig:Qhump} clearly shows that by performing perturbation on this matrix we are able to recover normal, inverse, or even humped curves. $Q$ seems to be best suited in determining large impacts on the shape of the yield curve, whereas $M$ seems to be suitable for smaller adjustments. Another fact that should be kept in mind, is that the interplay between the parameters is influenced by the different choice of the starting value of the process. 
\begin{figure}[H]
  \centering
  \subfloat[$M_{11}$]{\label{fig:M11h}\includegraphics[scale=0.04]{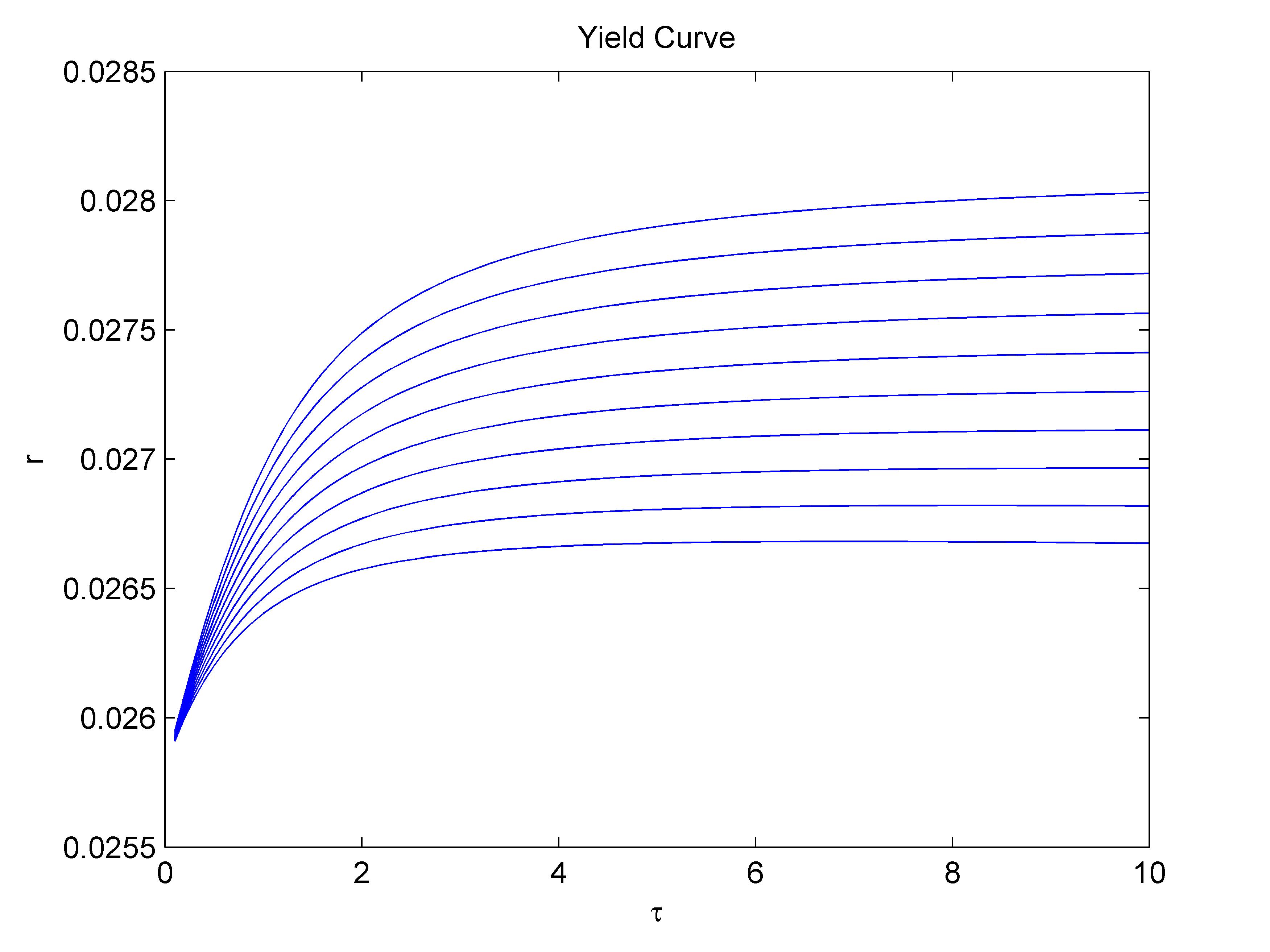}}
 \subfloat[$M_{12}$]{\label{fig:M12h}\includegraphics[scale=0.04]{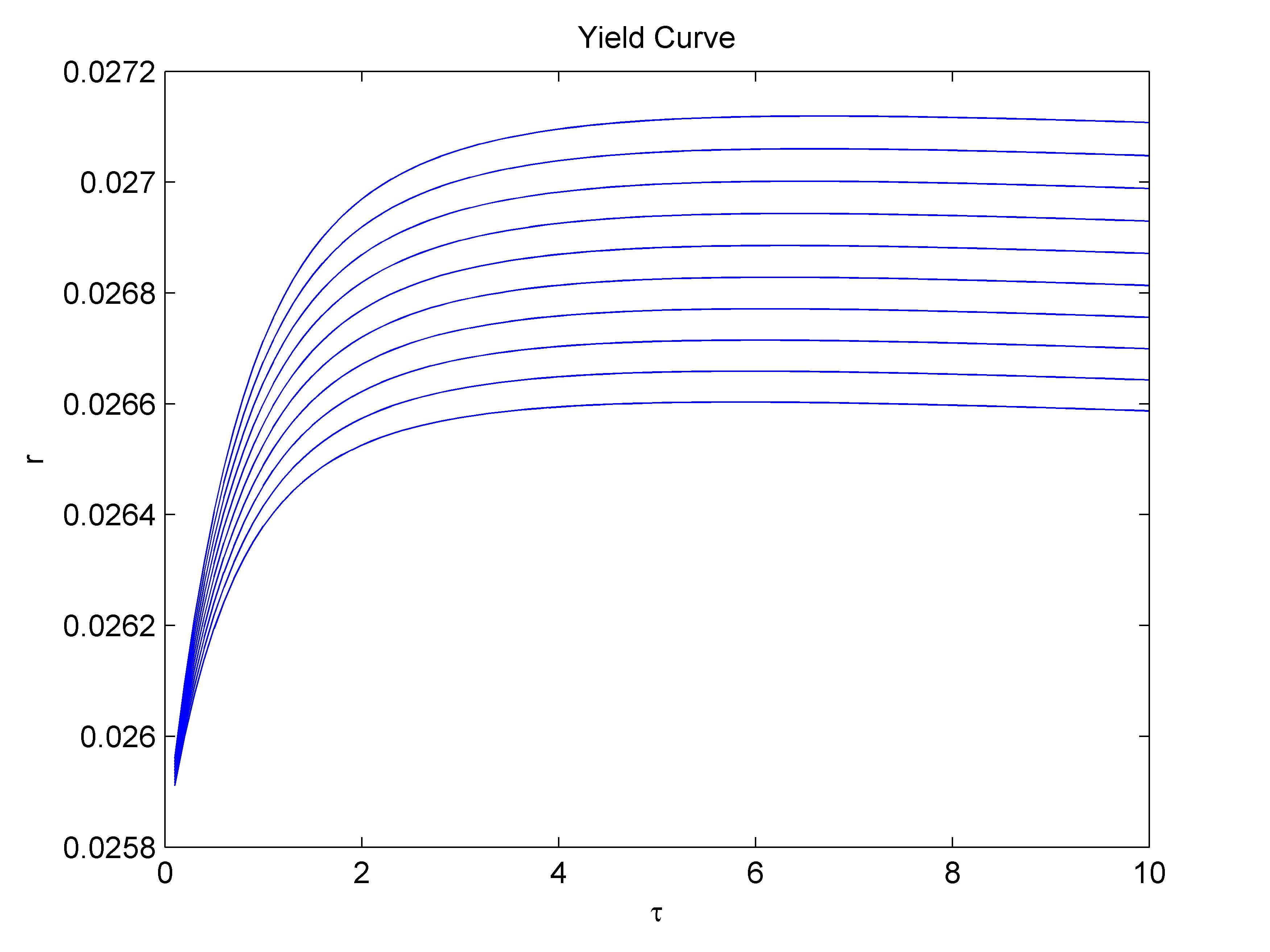}}\\
 \subfloat[$M_{21}$]{\label{fig:M21h}\includegraphics[scale=0.04]{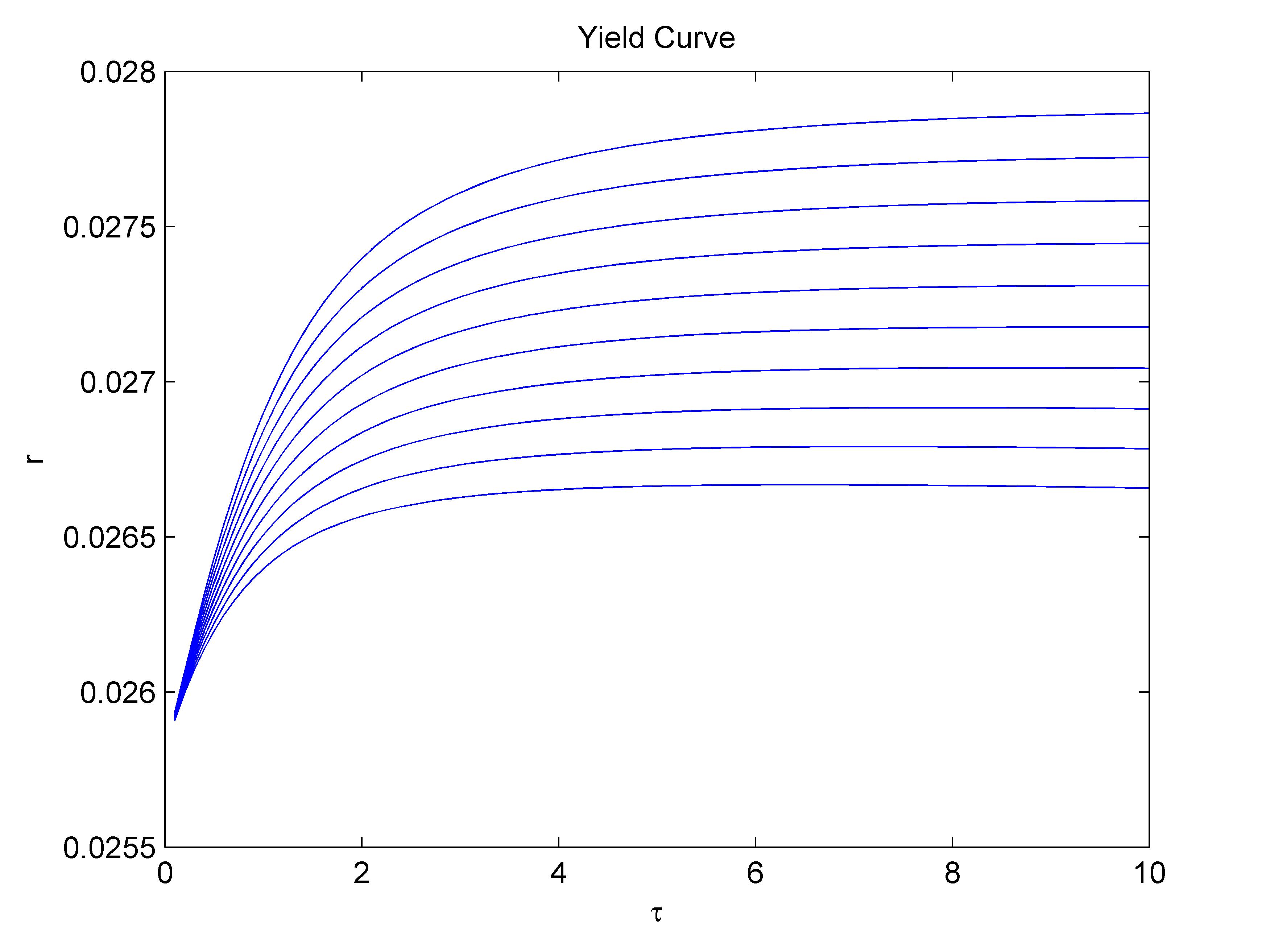}}
 \subfloat[$M_{22}$]{\label{fig:M22h}\includegraphics[scale=0.04]{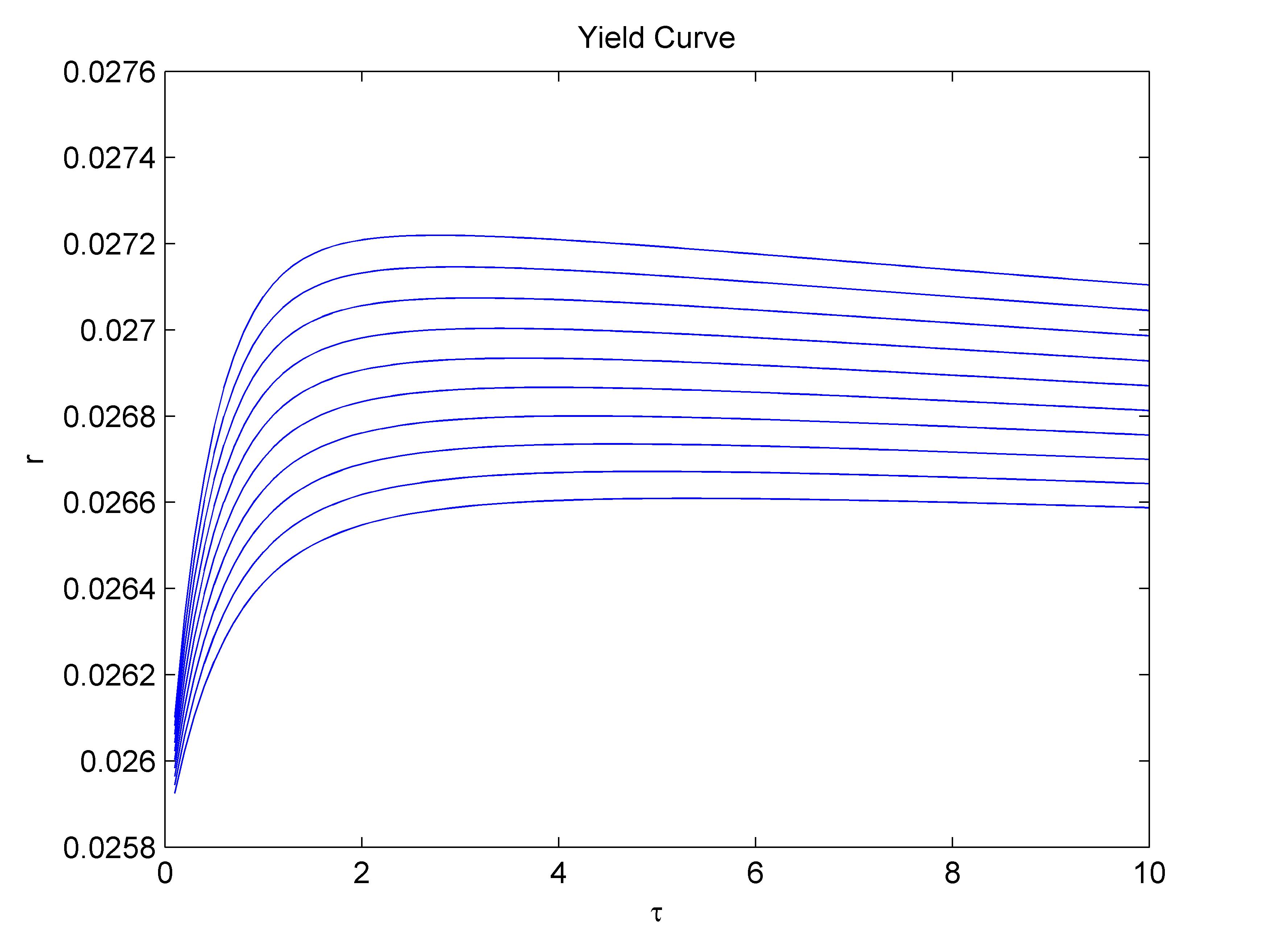}}
  \caption{We increase the starting value of the process so as to get a humped yield curve. In this figure we show the effect of a perturbation of the single elements of the matrix $M$. We use a sequence of numbers $\eta=0.01:0.01:0.1$ and add $\eta_i$ to one of the elements of $M$ while leaving the other elements unchanged. When we add $\eta_i$ to the elements on the main diagonal the yield curve is shifted upwards. The same happens with off-diagonal elements.}
  \label{fig:Mhump}
\end{figure}

\begin{figure}[H]
  \centering
  \subfloat[$Q_{11}$]{\label{fig:Q11h}\includegraphics[scale=0.04]{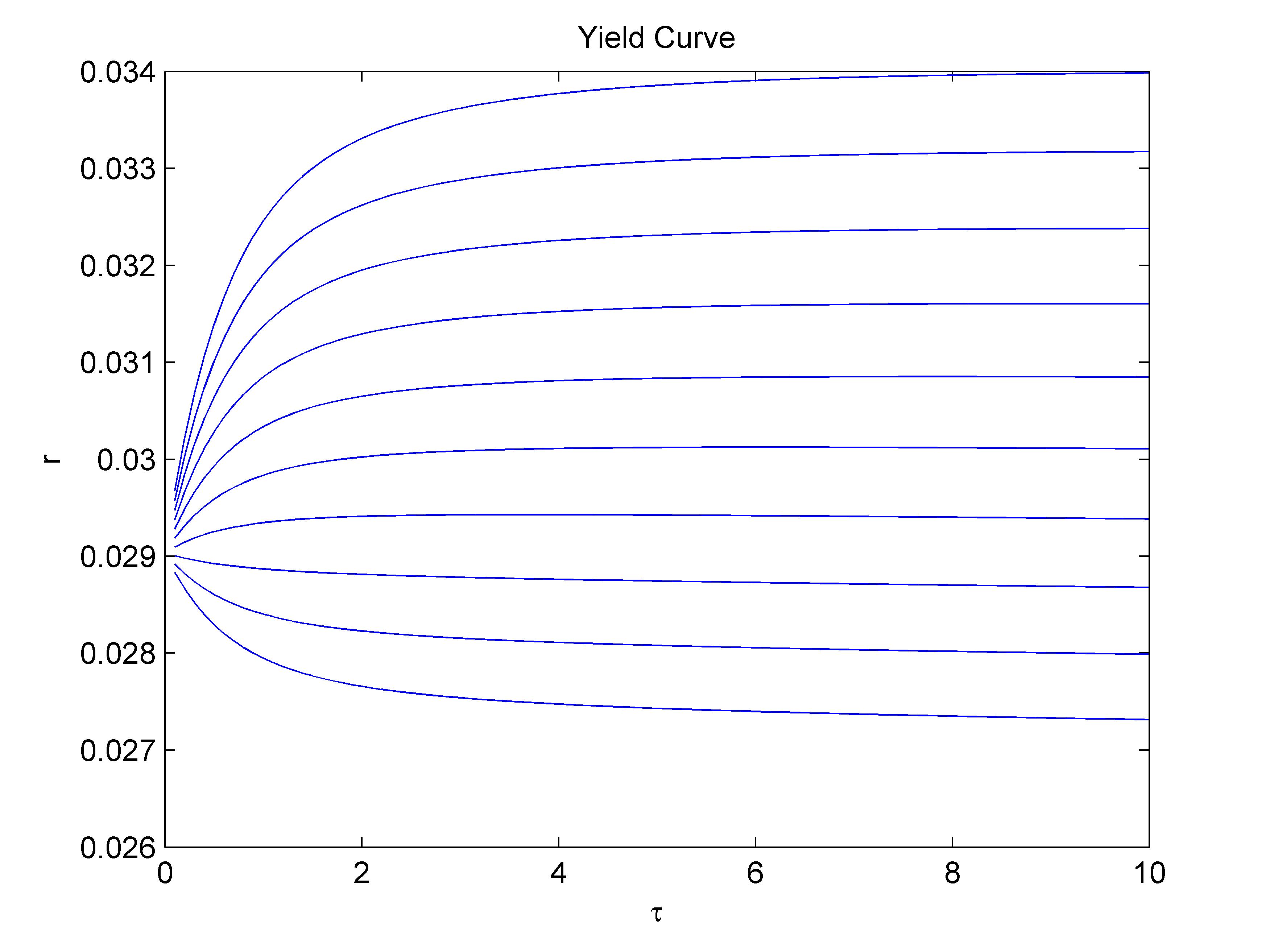}}
 \subfloat[$Q_{12}$]{\label{fig:Q12h}\includegraphics[scale=0.04]{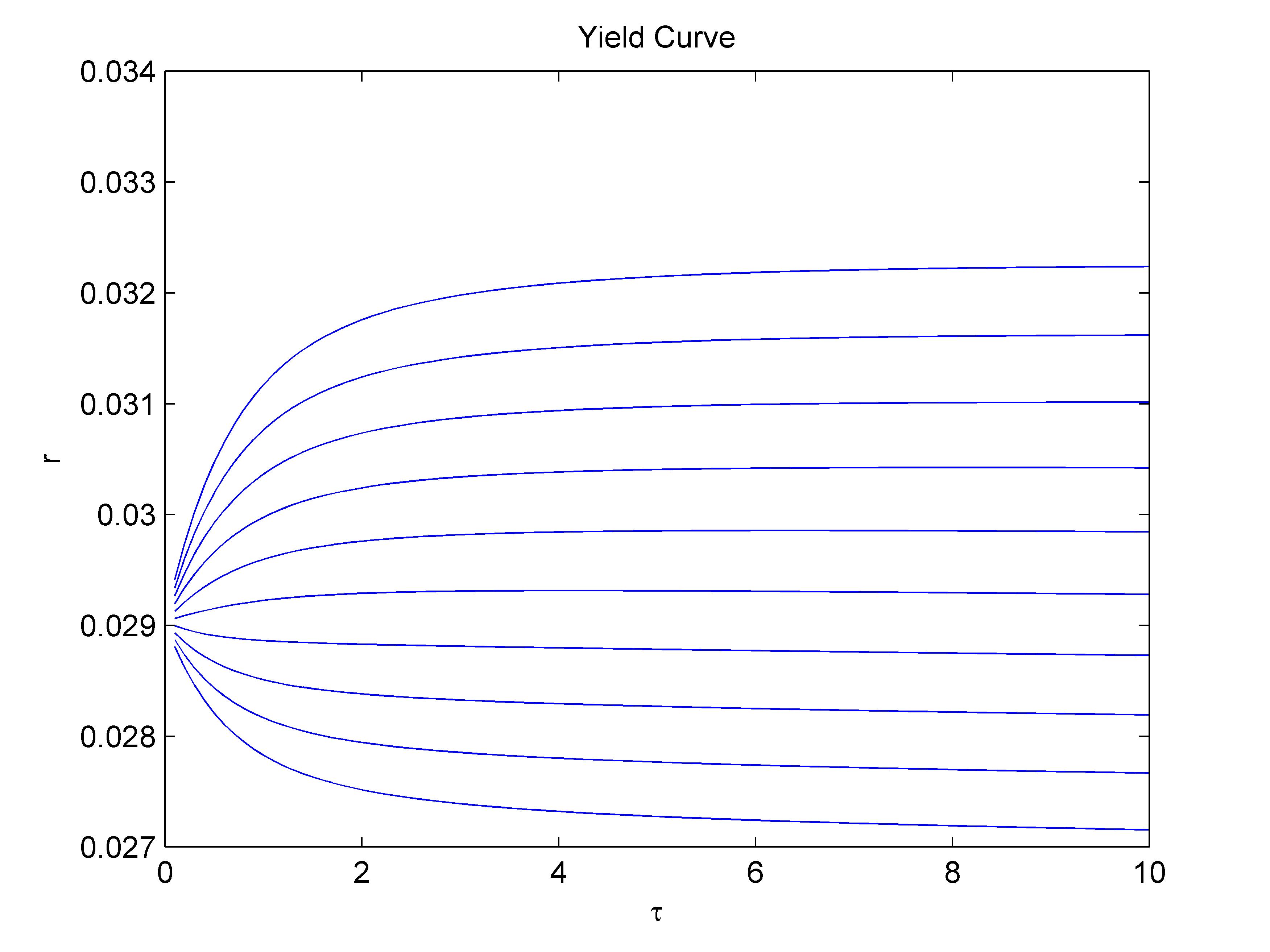}}\\
 \subfloat[$Q_{21}$]{\label{fig:Q21h}\includegraphics[scale=0.04]{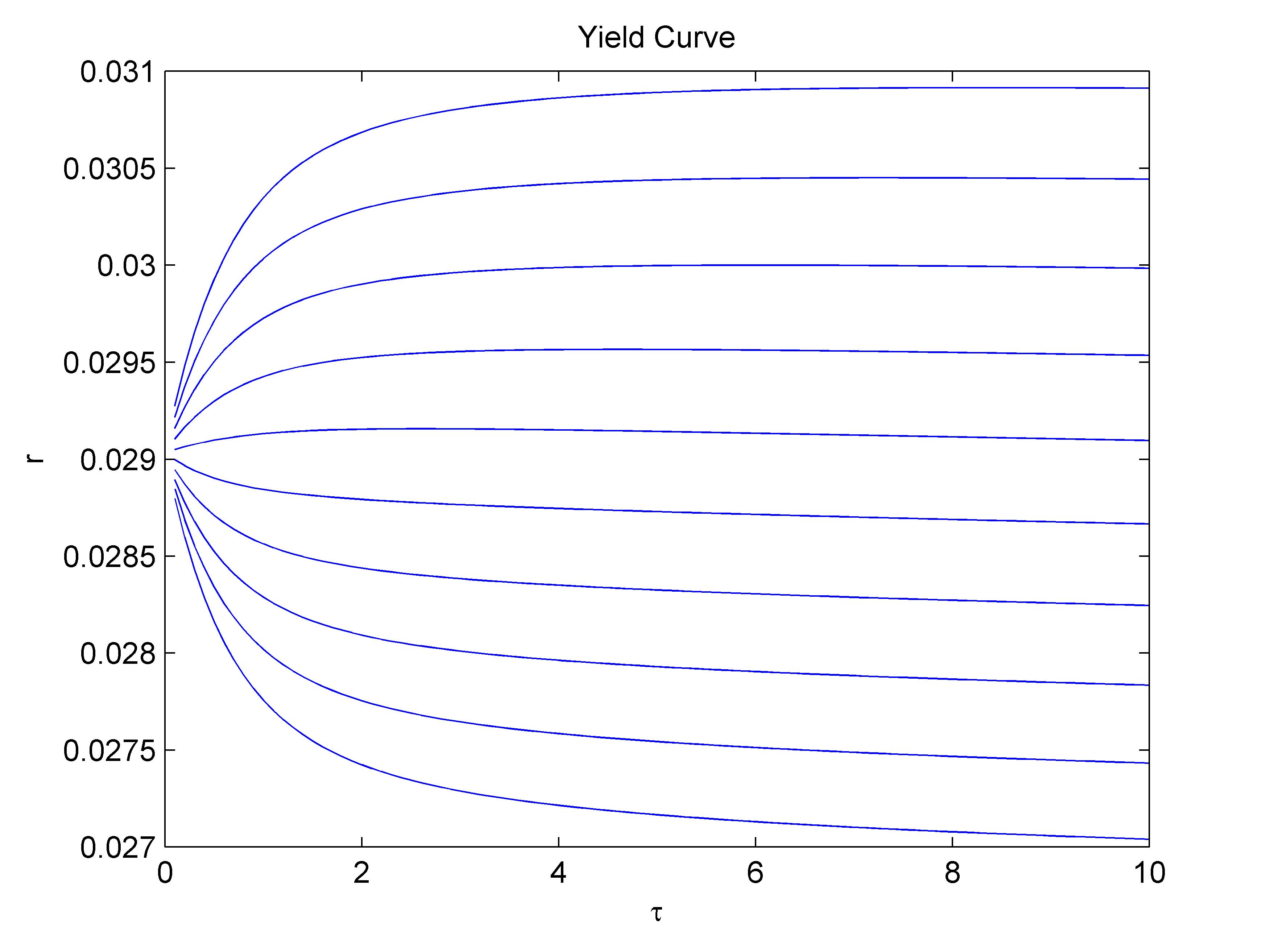}}
 \subfloat[$Q_{22}$]{\label{fig:Q22h}\includegraphics[scale=0.04]{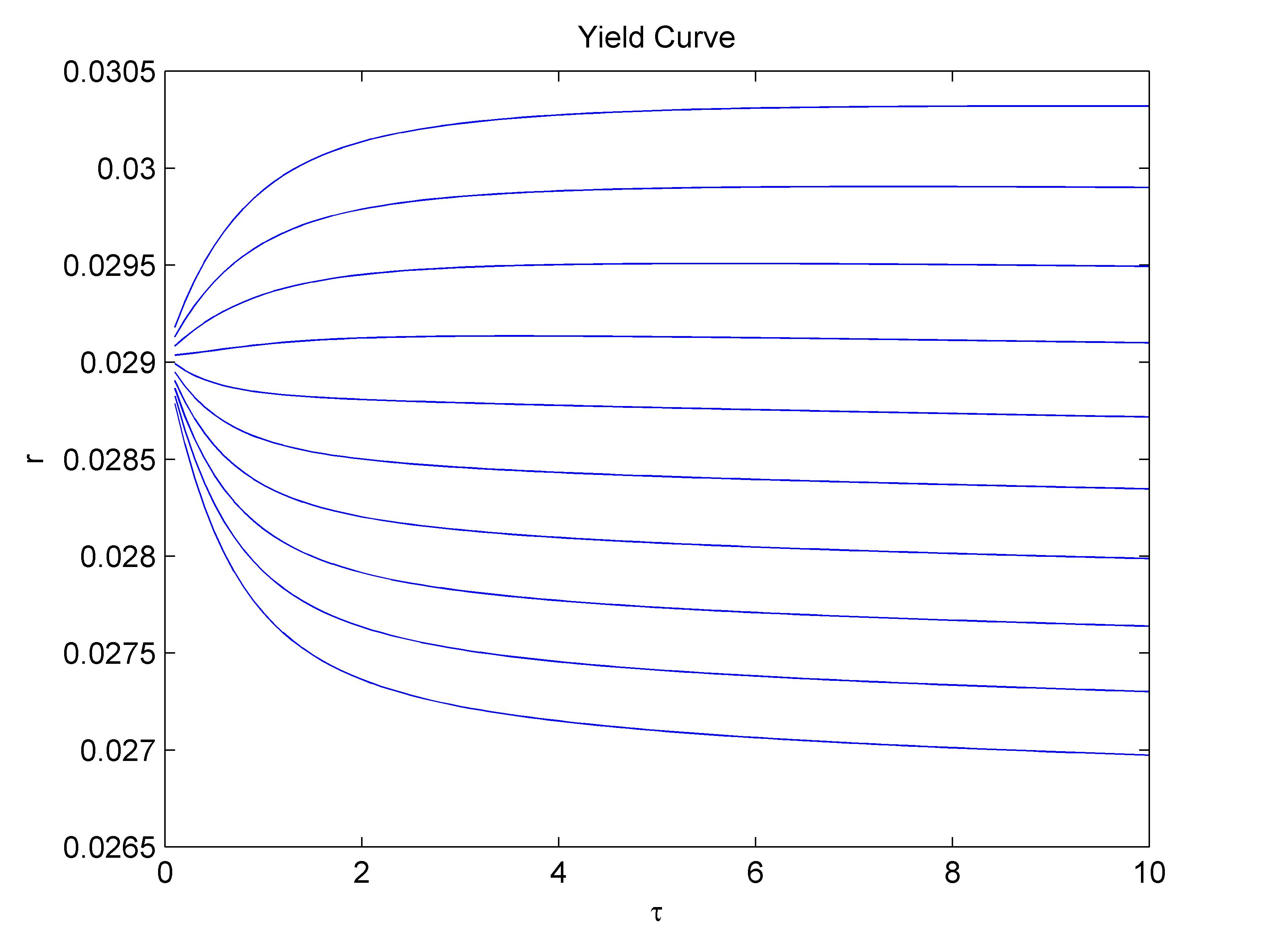}}
  \caption{We increase the starting value of the process so as to get a humped yield curve. In this figure we show the effect of a perturbation of the single elements of the matrix $Q$. We use a sequence of numbers $\eta=0.01:0.01:0.1$ and add $\eta_i$ to one of the elements of $Q$ while leaving the other elements unchanged. When we add $\eta_i$ to the elements on the main diagonal the yield curve is shifted upwards. The same happens with off-diagonal elements.}
  \label{fig:Qhump}
\end{figure}

\section{Conclusions}
In this paper we provided a set of sufficient conditions ensuring that the Wishart short rate model produces certain yield curve shapes. In particular we are able to ensure that the yield curve is normal, inverse or humped. We believe that this set of sufficient conditions may be of interest in a calibration setting, since we now know the constraints we should impose in order to replicate the shape of the yield curve that we observe on the market. A further line of research may be given by the study of yield curve shapes in the fully general model that was presented in Section 1.

\end{document}